\documentclass[11pt]{article}
\usepackage[T1]{fontenc}
\usepackage{amsfonts}
\usepackage{amsmath}
\usepackage{amssymb}
\usepackage{amsthm}
\usepackage{bbm}
\usepackage{bm}
\usepackage{mathrsfs}
\usepackage{verbatim}
\usepackage{setspace}
\usepackage{color}
\usepackage{pdfsync}
\usepackage{enumitem}

\theoremstyle{plain}
\newtheorem{theorem}{Theorem}[section]
\newtheorem{proposition}[theorem]{Proposition}
\newtheorem{lemma}[theorem]{Lemma}

\theoremstyle{definition}
\newtheorem{definition}[theorem]{Definition}
\newtheorem{remark}[theorem]{Remark}

\newtheorem{example}[theorem]{Example}

\newtheorem{assumption}[theorem]{Assumption}

\theoremstyle{remark}

\renewenvironment{thebibliography}[1]{%
\begin{oldthebibliography}{#1}%
\setlength{\baselineskip}{.9em}
\linespread{1}
\small
\setlength{\parskip}{0.3ex}%
\setlength{\itemsep}{.5em}%
}%
{%
\end{oldthebibliography}%
}
\newcommand{\eps}{\varepsilon}

\newcommand{\F}{\mathbb{F}}

\newcommand{\N}{\mathbb{N}}

\newcommand{\R}{\mathbb{R}}
\renewcommand{\S}{\mathbb{S}}

\newcommand{\cA}{\mathcal{A}}

\newcommand{\cE}{\mathcal{E}}
\newcommand{\cF}{\mathcal{F}}

\newcommand{\cK}{\mathcal{K}}
\newcommand{\cL}{\mathcal{L}}

\newcommand{\cS}{\mathcal{S}}

\newcommand{\sC}{\mathscr{C}}

\newcommand{\fP}{\mathfrak{P}}
\newcommand{\fM}{\mathfrak{M}}
\newcommand{\fg}{\mathfrak{g}}

\DeclareMathOperator*{\argmin}{arg\, min}
\DeclareMathOperator*{\argmax}{arg\, max}

\numberwithin{equation}{section}

\usepackage[pdfborder={0 0 0}]{hyperref}
\hypersetup{
  urlcolor = black,
  pdfauthor = {Ariel Neufeld, Marcel Nutz},
  pdfkeywords = {Utility maximization; Knightian uncertainty; Nonlinear Levy process
},
  pdftitle = {Robust Utility Maximization with Levy Processes},
  pdfsubject = {Robust Utility Maximization with Levy Processes},
  pdfpagemode = UseNone
}

\begin{document}

\title{\vspace{-0em}
\mbox{Robust Utility Maximization with L{\'e}vy Processes}
\date{\today}
\author{
  Ariel Neufeld%
  \thanks{
  Department of Mathematics, ETH Zurich, \texttt{ariel.neufeld@math.ethz.ch}.
  Financial support by Swiss National Science Foundation Grant PDFMP2-137147/1 is gratefully acknowledged.
  }
  \and
  Marcel Nutz%
  \thanks{
  Departments of Statistics and Mathematics, Columbia University, New York, \texttt{mnutz@columbia.edu}. Financial support by NSF Grants DMS-1512900 and DMS-1208985 is gratefully acknowledged. The authors thank Kostas Kardaras for enlightening discussions and two anonymous referees for helpful comments.
  }
 }
}
\maketitle \vspace{-1.2em}

\begin{abstract}
We study a robust portfolio optimization problem under model uncertainty for an investor with logarithmic or power utility. The uncertainty is specified by a set of possible L{\'e}vy triplets; that is, possible instantaneous drift, volatility and jump characteristics of the price process. We show that an optimal investment strategy exists and compute it in semi-closed form. Moreover, we provide a saddle point analysis describing a worst-case model.
\end{abstract}

\vspace{.9em}

{\small

}

\section{Introduction}\label{se:intro}

We study a robust utility maximization problem of the form
\begin{equation}\label{eq:valueIntro}
  \sup_{\pi}\inf_{P} E^{P}[U(W^{\pi}_{T})]
\end{equation}
in a continuous-time financial market with jumps. Here~$W^{\pi}_{T}$ is the wealth at time~$T$ resulting from investing in $d$ stocks according to the trading strategy~$\pi$ and $U$ is either the logarithmic utility $U(x)=\log(x)$ or a power utility $U(x)=\frac{1}{p}x^{p}$ for some $p\in(-\infty,0)\cup (0,1)$. The infimum is taken over a class~$\fP$ of possible models~$P$ for the dynamics of the log-price processes of the stocks. More precisely, the model uncertainty is parametrized by a set~$\Theta$ of L\'evy triplets $(b,c,F)$ and then~$\fP$ consists of all semimartingale laws $P$ such that the associated differential characteristics $(b^{P}_{t},c^{P}_{t},F^{P}_{t})$ take values in $\Theta$, $P\times dt$-a.e. In particular, $\fP$ includes all L\'evy processes with triplet in $\Theta$, but unless $\Theta$ is a singleton, $\fP$ will also contain many laws for which $(b^{P}_{t},c^{P}_{t},F^{P}_{t})$ are time-dependent and random. Thus, our setup describes uncertainty about drift, volatility and jumps over a class of fairly general models. 

Our first main result shows that an optimal trading strategy $\hat\pi$ exists for~\eqref{eq:valueIntro}. This strategy is of the constant-proportion type; that is, a constant fraction of the current wealth is invested in each stock. We compute this fraction in semi-closed form, so that the impact of model uncertainty can be readily read off; cf.\ Theorem~\ref{thm}. Thus, our specification of model uncertainty retains much of the tractability of the classical utility maximization problem for exponential L\'evy processes. This is noteworthy for the power utility as~$\fP$ contains models~$P$ that are not L\'evy and in which the classical power utility investor is not myopic. Moreover, while the classical $\log$ utility investor is myopic in any given semimartingale model, this property generally fails in robust problems, due to the nonlinearity caused by the infimum---retaining the myopic feature is specific to the setup chosen here, and in particular the (nonlinear) i.i.d.\ property of the increments of the log-prices under the nonlinear expectation $\inf_{P\in\fP} E^{P}[\,\cdot\,]$ in the sense of~\cite{HuPeng.09levy,NeufeldNutz.13b}.

Under a compactness condition on $\Theta$, we also show the existence of a worst-case model~$\hat{P}\in\fP$. This model is a L\'evy law and the corresponding L\'evy triplet $(\hat b,\hat c, \hat F)$ is computed in semi-closed form. More precisely, our second main result yields a saddle point $(\hat P,\hat\pi)$ for the problem~\eqref{eq:valueIntro} which may be seen as a two player zero-sum game. The strategy $\hat\pi$ and the triplet $(\hat b,\hat c, \hat F)$ are characterized as a saddle point of a deterministic function; cf.\ Theorem~\ref{thm-compact}. The fact that $\hat{P}$ is a L\'evy model may be compared with option pricing in the Uncertain Volatility Model, where in general the worst-case model is a non-L\'evy law unless the option is convex or concave.

Mathematically, our method of proof follows the local-to-global paradigm. That is, we first derive versions of our main results for a ``local'' optimization problem that plays the role of a Bellman-Isaacs operator. The passage to the global results is relatively direct in the logarithmic case, because the log investor is myopic in every model $P\in\fP$. For the power utility, this fails and thus the optimal strategy and expected utility for a fixed $P$ cannot be expressed in a simple way. However, we shall see that the worst case over all L\'evy laws already corresponds to the worst case over all $P\in\fP$. The key tool for this is a martingale argument; cf.\ Lemma~\ref{le:expo-martingale}.

Within the rich  literature on the portfolio optimization problem, going back to~\cite{Merton.69,Samuelson.69}, the present paper follows a branch which focuses on obtaining explicit or semi-explicit expressions for optimal portfolios. Essentially, this is possible only for isoelastic utility functions; moreover, a tractable model for the stock prices is required. While~\cite{Merton.69} provides the closed-form solution in the classical Black--Scholes model, a semi-explicit optimizer is still available for exponential L\'evy processes; see, e.g., \cite{Kallsen.00, Nutz.09c}. Semi-explicit solutions are also available for certain stochastic volatility models such as Heston's; see, e.g., \cite{KallsenMuhleKarbe.08, StoikovZariphopoulou.05}, among many others. The main merit of these solutions is to yield insight into how the presence of a specific phenomenon, such as stochastic volatility or jumps, may influence the choice of an investment strategy in comparison to more classical models. Here, our purpose is to study specifically the influence of model uncertainty.

Much of the literature on robust utility maximization in mathematical finance, starting with \cite{Quenez.04, Schied.06}, assumes that the set $\fP$ of models is dominated by a reference measure $P_*$. This assumption leads to a setting where volatilities and jump sizes are perfectly known, only drifts may be uncertain. By contrast, we are interested in uncertainty about all these three components, so that $\fP$ is nondominated. 
In general, the existence of optimal portfolios is known only in discrete time \cite{Nutz.13util}; however, \cite{DenisKervarec.13} establishes a minimax result and the existence of a worst-case measure in a setup where prices have continuous paths and the utility function is bounded. Continuous-time models with jumps have not been studied in the extant literature.

The early contribution \cite{TalayZheng.02} studies a class of related model risk management problems and shows that the lower value function ($\inf\sup$) solves a nonlinear PDE (these problems, however, do not admit a saddle point in general). In a setting closer to ours but again without jumps, \cite{MatoussiPossamaiZhou.12utility} obtains existence in a problem where $U$ is an isoelastic utility function, volatility is uncertain (within an interval) but the drift is known, by considering an associated second order backward stochastic differential equation. On the other hand, \cite{TevzadzeToronjadzeUzunashvili.13} studies the Hamilton--Jacobi--Bellman--Isaacs PDE related to the robust utility maximization problem in a diffusion model with a non-tradable factor and miss-specified drift and volatility coefficients for the traded asset; here a saddle point can be found after a randomization.
A model with several uncorrelated stocks, where drift, interest rate and volatility are uncertain within a specific parametrization, is considered in~\cite{LinRiedel.14}. A saddle point is found and analyzed, again by dynamic programming arguments. Recently, \cite{BiaginiPinar.15} also constructs a saddle point in a setting where the uncertainty in the drift may depend on the realization of the volatility in a specific way. Finally, \cite{FouquePunWong.14} considers a stochastic volatility model with uncertain correlation (but known drift) and describes an asymptotic closed-form solution. 
In the present paper, our main contribution is to exhibit and solve a problem that includes uncertainty about fairly general models while remaining very tractable.

The remainder of this paper is organized as follows. In Section~\ref{se:theProblem}, we specify our model and the optimization problem in detail, and we state our main results. Section~\ref{se:localAnalysis} contains the analysis of the local optimization problem. In Section~\ref{se:proof-log}, we give the proofs of the main results for the logarithmic utility, whereas Section~\ref{se:proof-pow} presents the proofs for power utility.

\section{The Optimization Problem}\label{se:theProblem}

\subsection{Setup for Model Uncertainty}

We fix the dimension $d \in \N$ and let $\Omega=D_0(\R_+,\R^d)$ be the space of all c\`adl\`ag paths $\omega=(\omega_t)_{t\geq 0}$ starting at $0\in\R^{d}$. We equip $\Omega$ 
with the Skorohod topology and
the corresponding Borel $\sigma$-field $\cF$. Moreover, we denote by $X = (X_t)_{t\geq 0}$
the canonical process $X_t(\omega) = \omega_t$, by $\F$ = $(\cF_t)_{t\geq 0}$ the (raw) filtration
generated by $X$, and by $\fP(\Omega)$ the Polish space of all probability
measures on $\Omega$. We also fix the time horizon $T\in (0,\infty)$. 

The uncertainty about drift, volatility and jumps is parametrized by a nonempty set 
$$
  \Theta\subseteq \R^d\times \S^d_+ \times \cL,
$$
where $\cL$ is the set of L\'evy measures; i.e., the set of all measures $F$ on $\R^{d}$ that satisfy $\int_{\R^d} |z|^2 \wedge 1 \, F(dz)<\infty$ and $F(\{0\})=0$.
We write 
\begin{equation*}
  \cL_{\Theta}= \{F\in\cL\,|\, (b,c,F)\in\Theta\}
\end{equation*}
for the projection of $\Theta$ onto $\cL$. The class of models to be considered is represented by the set $\fP$ of semimartingale laws such that the differential characteristics of the canonical process $X$ take values in $\Theta$. More precisely, let
\begin{equation*}
    \fP_{sem}=\big\{P\in \fP(\Omega)\,\big|\, \mbox{$X$ is a semimartingale on }(\Omega,\cF,\F,P)\big\}
\end{equation*}
be the set of all semimartingale laws, denote by $(B^P,C^P,\nu^P)$ the predicable characteristics of $X$ under $P$ with respect to a fixed truncation function $h$, and let
\begin{equation*}
   \fP^{ac}_{sem}= \big\{P\in \fP_{sem}\,\big|\, \mbox{$(B^P,C^P,\nu^P)\ll dt$, $P$-a.s.} \big\}
\end{equation*}
be the set of semimartingale laws with absolutely continuous characteristics (with respect to the Lebesgue measure $dt$).
Given such a triplet $(B^P,C^P,\nu^P)$, the corresponding derivatives (defined $dt$-a.e.) are called the differential characteristics of $X$ and denoted by $(b^P,c^P,F^P)$. Our set $\fP$ of possible laws is then given by
\begin{equation*}
  \fP=\big\{P\in \fP^{ac}_{sem}\,\big|\, (b^P,c^P,F^P)\in \Theta, \, P\otimes dt\mbox{-a.e.}\big\}.
\end{equation*}
The canonical process $X$, considered under the set $\fP$, can be seen as a nonlinear L\'evy process in the sense of~\cite{NeufeldNutz.13b}. Finally, let us denote by
\begin{equation*}
  \fP_{L}=\big\{P\in \fP\,\big|\, \mbox{$X$ is a L\'evy process under $P$}\big\}
\end{equation*}
the set of all L\'evy laws in $\fP$. Thus, there is a one-to-one correspondence between $\fP_{L}$ and the set $\Theta$ of L\'evy triplets, whereas the set $\fP$ is in general much larger than $\fP_{L}$.

\subsection{Utility and Constraints} 

To model the preferences of the investor, we consider the logarithmic and the power utility functions on $(0,\infty)$; i.e.,
\[
 U(x)= \log(x) \  \quad \mbox{and }\quad U(x)=\frac{1}{p} x^p \  \mbox{ for } p\in(-\infty,0)\cup (0,1).
\]
As usual, we set $U(0):=\lim\limits_{x\to 0} U(x)$ and $U(\infty):=\lim\limits_{x\to \infty} U(x)$.

Our investor is endowed with a deterministic initial capital $x_0>0$ and chooses a trading strategy $\pi$; that is, a predictable $\R^d$-valued process which is $X$-integrable under all $P\in \fP$. Here the canonical process $X$ represents the returns of the (discounted) stock prices and thus the $i$th component of $\pi$ is interpreted as the \emph{proportion} of current wealth invested in the $i$th stock. Under any $P\in \fP$, the corresponding wealth process $W^{\pi}$ is given by the stochastic exponential
\[
  W^{\pi}=x_0\,\cE\Big( \int \pi \, dX\Big). 
\]
The portfolio is subject to a no-bankruptcy constraint that can be described by the set of  \emph{natural constraints},
\begin{equation*}
\sC^0 := \bigcap_{F \in \cL_{\Theta}} \big\{ y \in \R^d \, \big| \, F[z \in \R^d \, | \, y^\top z< -1]=0 \big\}.
\end{equation*}
Indeed, a strategy $\pi$ with values in $\sC^0$ satisfies $\pi^\top \Delta X\geq-1$ $P$-a.s.\ for all $P \in \fP$ and this is in turn equivalent to $W^\pi\geq0$ $P$-a.s.\ for all $P \in \fP$. For later use, we note that $\sC^0$ is a closed, convex subset of $\R^{d}$ that contains the origin.

In addition to the natural constraints, we may impose further constraints such as no-shortselling on the investor. These constraints are modeled by an arbitrary closed, convex set $\sC\subseteq \R^d$ containing the origin.

The set $\cA$ of \emph{admissible strategies} is the collection of all strategies $\pi$ such that $\pi_{t}(\omega)\in \sC\cap \sC^0$ for all $(\omega,t)\in [\![0,T]\!]$ and $U(W^{\pi}_T)>-\infty$ $P$-a.s.\ for all $P \in \fP$. The second condition is for notational convenience: if $U(W^{\pi}_T)=-\infty$ with positive probability for some $P \in \fP$, then $\pi$ is not relevant for our optimization problem. Note that nothing is being excluded for the power utility with $p\in(0,1)$, whereas in the other cases we have $U(0)=-\infty$ and thus $\pi \in \cA$ implies $W^{\pi}>0$ $P$-a.s.\ for all $P \in \fP$; cf.\ \cite[Theorem~I.4.61, p.\,59]{JacodShiryaev.03}.
The value function of our robust utility maximization problem is 
\begin{equation}\label{eq:utilityproblem}
  u(x_0):=\sup_{\pi \in \cA} \inf_{P \in \fP} E^P\big[U(W_T^{\pi})\big].
\end{equation}
Here and below, we define the expectation for any measurable function with values in $\overline{\R}=[-\infty,\infty]$, using the convention $\infty - \infty = -\infty$. We say that the robust utility maximization problem is \emph{finite}  if $u(x_0)<\infty$. Under this condition, we call $\pi\in \cA$ \emph{optimal} if it attains the supremum in~\eqref{eq:utilityproblem}.

\subsection{Main Results}

We recall that $U$ stands for either $U(x)= \log(x)$ or $U(x)=\frac{1}{p} x^p$ with $p\in(-\infty,0)\cup (0,1)$. For convenience of notation, $p=0$ will refer to the logarithmic case in what follows. The subsequent conditions are in force for the remainder of the paper.

\begin{assumption}\label{assumpt}
 $ $
\begin{enumerate}
\item The set $\sC \cap \sC^0\subseteq \R^d$ is compact.
\item The set $\Theta\subseteq \R^d\times \S^d_+ \times \cL$ is convex and satisfies $\cK<\infty$, where
  \[
     \mathcal{K} :=\sup\limits_{(b,c,F)\in \Theta} \begin{cases} 
      |b| + |c| + \int  |z|^2 \wedge \log(1+|z|)  \, F(dz)    &\mbox{if } p=0,
     \\ 
      |b| + |c| + \int  |z|^2 \wedge |z|^{p(1+\varepsilon)} \, F(dz)   &\mbox{if }  p\in (0,1), \\
          |b| + |c| + \int  |z|^2 \wedge 1 \, F(dz)   &\mbox{if }  p<0.
         \end{cases} 
       \]
      In the case $p\in (0,1)$, we have fixed an arbitrarily small constant $\varepsilon>0$ in the above definition of $\cK$.
       
\end{enumerate}
\end{assumption}

\begin{remark}
(i) The compactness assumption on $\sC \cap \sC^0$ is \emph{not} very restrictive for the cases of our interest: in the presence of jumps, the set $\sC^0$ is typically compact, and then the assumption holds even in the unconstrained case $\sC=\R^{d}$. Indeed, let $d=1$ for simplicity of notation. As soon as the the jumps of $X$ are unbounded from above, for at least one $P\in\fP$, and not bounded away from $-1$, for at least one $P\in\fP$, then $\sC^0\subseteq [0,1]$ and $\sC \cap \sC^0$ is necessarily compact.

The non-compact case can also be analyzed but leads to technical complications that are not of specific interest to our robust problem. These complications are well-studied in the classical case; cf.\ \cite{Kardaras.09,Nutz.09c}. The non-compact case is of interest, in particular, in the Uncertain Volatility Model without jumps.

(ii) The second condition in Assumption~\ref{assumpt} will guarantee, in particular, the finiteness of the utility maximization problem. When $p<0$, no specific L\'evy triplet is excluded as any L\'evy measure satisfies $\int |z|^2 \wedge 1 \, F(dz) <\infty$. When $p\geq 0$, a sufficient condition is that the L\'evy process has integrable jumps, which is equivalent to $\int |z|^2 \wedge |z| \, F(dz) <\infty$.
\end{remark}

\begin{definition}
	Let $(b,c,F)\in \R^d\times \S^d_+ \times \cL_\Theta$. For $y\in\sC^0$, we define
	\begin{equation}\label{eq:fct_g-bcF}
	g^{(b,c,F)}(y):= y^\top b +\frac{p-1}{2} \,y^\top c y + \int_{\R^d} I_{y}(z)\, F(dz), 
	\end{equation}
	where
	\begin{equation*}
	I_{y}(z):=\left\{ \begin{array}{ll}
	 \log(1+y^\top z)-y^\top h(z) \ &\mbox{if } \ p=0, \\
	     &\\
	        p^{-1}(1+y^\top z)^p -p^{-1}-y^\top h(z) \ &\mbox{if } \ p\neq 0.
	\end{array}\right.
  \end{equation*}
We shall see later that $g^{(b,c,F)}$ is a well-defined concave function with values in $[-\infty,\infty)$.
\end{definition}

Our first main result states that an optimal strategy exists; moreover, it is given by a constant proportion that can be described in terms of the function~$g$. We recall that Assumption~\ref{assumpt} is in force.

\begin{theorem}[Optimal Strategy]\label{thm}
$ $
\begin{enumerate}
\item  The robust utility maximization problem is finite and
\begin{equation}\label{eq:global-max-thm}
\sup_{\pi \in \cA} \inf_{P\in \fP} E^P\big[U(W_T^{\pi})\big]
= \inf_{P\in \fP} \sup_{\pi \in \cA} E^P\big[U(W_T^{\pi})\big].
\end{equation}
\item  There exists an optimal strategy which is constant. More precisely, the finite-dimensional problem
\begin{equation*}
\argmax _{y \in \sC\cap \sC^0} \inf_{(b,c,F)\in \Theta} g^{(b,c,F)}(y)
\end{equation*}
has at least one solution. Any solution $\hat{y}$, 
seen as a constant process,  is in $\cA$ and defines an optimal strategy; i.e.,
\begin{equation*}%
\inf_{P \in \fP}E^P\big[U(W_T^{\hat{y}})\big]
= \sup_{\pi \in \cA} \inf_{P\in \fP} E^P\big[U(W_T^{\pi})\big],
\end{equation*}
and this value is equal to
\begin{equation*}
\left\{ \begin{array}{ll}
\log(x_0)+ \,T \inf\limits_{(b,c,F)\in \Theta} g^{(b,c,F)}(\hat{y}) \ & \mbox{if } \ p=0, \\
     &\\
        \frac{1}{p}x_0^p\exp\Big(p\,T \inf\limits_{(b,c,F)\in \Theta} g^{(b,c,F)}(\hat{y})\Big) \ &\mbox{if } \ p\neq 0.
\end{array}\right.
\end{equation*}
\item Conversely, any constant optimal strategy $\tilde{\pi}\in \cA$ is an element of 
\begin{equation*}
\argmax _{y \in \sC\cap \sC^0} \inf_{(b,c,F)\in \Theta} g^{(b,c,F)}(y).
\end{equation*}
\end{enumerate}
\end{theorem}

The robust utility maximization problem can be seen as a two player zero-sum game. Indeed, the minimax identity~\eqref{eq:global-max-thm} then states the existence of the value. Our second main result is a saddle point analysis of the game. For reference, let us recall that a point $(\hat{x},\hat{y})\in \mathcal{X}\times \mathcal{Y}$ in some product set is called a saddle point of the function
$f:\mathcal{X}\times \mathcal{Y} \to [-\infty,\infty]$ if
\begin{equation*}
 f(\hat{x},y)\leq f(\hat{x},\hat{y})\leq f(x,\hat{y}) \quad \mbox{for all}\quad  x\in \mathcal{X},\:y \in \mathcal{Y}.
\end{equation*}
Thus, $\hat{x}$ is the optimal response when the second player chooses $\hat{y}$, and vice versa. This is equivalent to the three assertions
\begin{enumerate}
\item $\sup_{y \in \mathcal{Y}} f(\hat{x},y)= \inf_{x \in \mathcal{X}} \sup_{y \in \mathcal{Y}} f(x,y)$,
\item  $\inf_{x \in \mathcal{X}} f(x,\hat{y})= \sup_{y \in \mathcal{Y}} \inf_{x \in \mathcal{X}} f(x,y)$,
\item $\sup_{y \in \mathcal{Y}} \inf_{x \in \mathcal{X}} f(x,y) = \inf_{x \in \mathcal{X}} \sup_{y \in \mathcal{Y}} f(x,y)$;
\end{enumerate}
that is, $\hat{x}$ and $\hat{y}$ solve the respective robust optimization problems, and the minimax identity holds.

To provide a saddle point analysis of the game, we need to introduce a topology on the set $\Theta$.
Recall first that the space $\fM^f$ of all finite measures on $\R^d$ is a Polish space under a metric $d_{\fM^f}$ which induces the weak convergence relative to $C_b(\R^d)$; cf.\ \cite[Theorem~8.9.4, p.\,213]{Bogachev.07volII}. This topology is the natural extension of the more customary weak convergence of probability measures. With any L\'evy measure $\mu\in\cL$ we can associate the finite measure
\begin{equation*}
  A \mapsto \int_A |x|^2 \wedge 1 \, \mu(dx), \ \ \ A \in \mathcal{B}(\R^d),
\end{equation*}
 denoted by $|x|^2 \wedge 1.\mu$. We can now define a metric $d_{\cL}$ via
\begin{equation*}
  d_{\cL}(\mu,\nu)= d_{\fM^f}\big(|x|^2 \wedge 1.\mu, |x|^2 \wedge 1.\nu\big), \quad \mu,\nu \in \cL,
\end{equation*}
and then $(\cL,d_{\cL})$ is a separable metric space; cf.\, \cite[Lemma~2.3]{NeufeldNutz.13a}. Moreover, the following version of Prohorov's theorem holds: a set $\cS\subseteq \cL$ is relatively compact if and only if 
  \begin{enumerate}
  \item $\sup_{F \in \cS} \int_{\R^d} |z|^2 \wedge 1 \, F(dz)<\infty$ and
  \item for any $\delta>0$ there exists a compact set $K_\delta \subseteq \R^d$ such that 
  \begin{equation*}
\sup_{F \in \cS} \int_{K^c_\delta} |z|^2 \wedge 1 \, F(dz) \leq \delta.
  \end{equation*}
  \end{enumerate}
  This is a consequence of~\cite[Theorem~1.12]{Prohorov.56}.
   
  Having defined a topology on $\cL$, we can equip $\R^d\times \S^d_+ \times \cL$ with the corresponding product topology and state our second main result; recall that $\fP_{L}$ denotes the set of all L\'evy laws in $\fP$.

\begin{theorem}[Saddle Point]\label{thm-compact}
  Let $\Theta\subseteq \R^d\times \S^d_+ \times \cL$ be compact.\begin{enumerate}
\item  The function $(P,\pi)\mapsto E^P\big[U(W^\pi_T)\big]$ has a saddle point on $\fP \times \cA$.

More precisely, 
the function $g^{(b,c,F)}(y)$ defined in \eqref{eq:fct_g-bcF} has a saddle point on $\Theta \times\sC\cap\sC^0$. If $\big((\hat{b},\hat{c}, \hat{F}),\hat{y}\big)$ is any such saddle point and $\hat{P} \in \fP_{L}$ denotes the L\'evy law with triplet $(\hat{b},\hat{c}, \hat{F})$, then $\hat{y} \in \cA$  and $(\hat{P},\hat{y})$ is a saddle point of $(P,\pi)\mapsto E^P\big[U(W^\pi_T)\big]$ on $\fP \times \cA$, and its value is 
\begin{equation*}
E^{\hat{P}}\big[U(W^{\hat{y}}_T)\big]=
\left\{ \begin{array}{ll}
\log(x_0) + T \, g^{(\hat{b},\hat{c}, \hat{F})}(\hat{y}) \ & \mbox{if } \ p=0, \\
     &\\
        \frac{1}{p}x_0^p\exp\big(p\,T \, g^{(\hat{b},\hat{c}, \hat{F})}(\hat{y})\big) \ &\mbox{if } \ p\neq 0.
\end{array}\right.
\end{equation*}
\item Conversely, if $(\tilde{P},\tilde{\pi})$ is a saddle point of  $(P,\pi)\mapsto E^P\big[U(W^\pi_T)\big]$ on $\fP\times\cA$, and   $\tilde{P}\in \fP_{L}$ and $\tilde{\pi}$ is constant, then $\big((\tilde{b},\tilde{c},\tilde{F}),\tilde{\pi}\big)$ is a saddle point of the function $g^{(b,c,F)}(y)$ on $ \Theta\times\sC\cap\sC^0$, where $(\tilde{b},\tilde{c},\tilde{F})$ is the L\'evy triplet of $\tilde{P}$.
\end{enumerate}
 \end{theorem}
 
 We remark that the worst-case model $\hat{P}$ is not unique in any meaningful way. For instance, if $\sC=[0,1]$ and $\Theta\subseteq \R_{-}\times[0,\infty)\times\{0\}$, then $(P,0)$ is a saddle point for any $P \in \fP_{L}$. Generally speaking, $P\mapsto \sup_{\pi} E^P\big[U(W^\pi_T)\big]$ is convex but not necessarily strictly convex, so that $\hat{P}$ may easily fail to be unique. On the other hand, $\hat{\pi}$ is unique in the sense that $W^{\hat\pi}$ is uniquely determined $\hat{P}$-a.s.
 
\begin{example}
  To give a simple example with an explicit solution, consider positive numbers $b_{min}\leq b_{max}$, $c_{min}\leq c_{max}$, $\lambda_{min}\leq \lambda_{max}$ and let
	\begin{equation*}
	\Theta:=\big\{(b,c,\lambda \delta_1) \,|\, b \in [b_{min},b_{max}], c\in [c_{min},c_{max}], \lambda\in [\lambda_{min},\lambda_{max}]\big\}
	\end{equation*}
	with respect to a truncation function $h$ satisfying $h(1)=1$. We consider the case $p=0$; moreover, $b_{min}-c_{max}-\lambda_{max}>0$ to avoid trivialities. Then, 
	\begin{equation*}
	\argmax_{y\in \sC\cap\sC^0} g^{(b,c,\lambda \delta_1)}(y)=\frac{(b-c-\lambda)+\sqrt{(b-c-\lambda)^2+4bc}}{2c}
	\end{equation*}
	provided that the constraint $\sC$ is large enough to contain that value. Moreover, 
	\begin{equation*}
	\argmin_{(b,c,F)\in \Theta} \sup_{y \in\sC\cap \sC^0} g^{(b,c,F)}(y)=(b_{min}, c_{max}, \lambda_{max}\,\delta_1)= :(\hat{b},\hat{c}, \hat{F}),
	\end{equation*}
	and we see that this triplet together with the strategy 
	$$
	\hat y := \frac{(b_{min}-c_{max}-\lambda_{max})+\sqrt{(b_{min}-c_{max}-\lambda_{max})^2+4b_{min}c_{max}}}{2c_{max}}
	$$
	forms the unique saddle point of
	$g^{(b,c,F)}(y)$
	on $\Theta\times \sC\cap\sC^0$.
\end{example}

\begin{remark}
   We may compare the situation of uncertainty over the set $\fP$ of semimartingale laws and the (much smaller) set $\fP_{L}$ of L\'evy laws. It follows from the proofs below that the value function, the optimal strategies and the saddle points in the main results are the same in both cases. This is in contrast, for example, to the situation of option pricing in the Uncertain Volatility Model, where the worst-case can be a non-L\'evy law.
\end{remark}

\section{The Local Analysis}\label{se:localAnalysis}

In this section we analyze the function $g^{(b,c,F)}$ defined in \eqref{eq:fct_g-bcF}; the results will be fundamental for the proofs of our main theorems. We set
\begin{equation}\label{eq:fct_g}
    g(y):= \inf_{(b,c,F) \in \Theta} g^{(b,c, F)}(y), \quad y\in\sC^{0}
\end{equation}
and recall that Assumption~\ref{assumpt} is in force.

\begin{lemma}\label{le:fct_g}
  Let $\theta=(b,c,F)\in \R^d\times \S^d_+ \times \cL_\Theta$. The function $g^{\theta}$ of~\eqref{eq:fct_g-bcF} is well-defined, proper, concave  and upper semicontinuous on $\sC^0$, with values in $[-\infty,\infty)$. The same holds for the function $g$ of~\eqref{eq:fct_g}. As a consequence,  $g^{\theta}$ and $g$ attain their maxima on $\sC\cap \sC^0$.
\end{lemma}
\begin{proof}
The first assertion follows directly from Assumption~\ref{assumpt} and the literature on classical utility maximization; cf.\ \cite[Section~5.1, p.\,182]{Kardaras.09} for $p=0$ and \cite[Lemma~5.3]{Nutz.09c} for $p \neq 0$. The remaining assertions are direct consequences.
\end{proof}

It will be useful to avoid the singularity of $g^{(b,c,F)}$ by considering the closed, convex sets
\begin{equation*}
\sC^0_{n}:= \bigcap_{F \in \cL_{\Theta}} \Big\{ y \in \R^d \, \Big| \, F\big[z \in \R^d \, \big| \, y^\top z< -1+\tfrac{1}{n}\big]=0 \Big\} \subseteq \sC^0
\end{equation*}
for $n\in\N$. We have the following approximation result.

\begin{lemma}\label{le:approx-g}
Let $\theta \in \R^d \times \S^d_+\times \cL_\Theta$ and let $\hat{y}^{\theta}$ be a maximizer of $y\mapsto g^{\theta}(y)$ on $\sC\cap \sC^0$; then
\begin{equation*}
\sup_{y \in \sC\cap \sC^0}g^{\theta}(y) 
= \lim\limits_{n \to \infty} \sup_{y \in \sC\cap \sC^0_n}g^{\theta}(y) 
= \lim\limits_{n \to \infty} g^{\theta}\big(\hat{y}^{\theta}_n\big)\quad\mbox{for}\quad\hat{y}^{\theta}_n:= (1-\tfrac{1}{n})\hat{y}^{\theta}.
\end{equation*}
Similarly, let $\hat{y}$ be a maximizer of $y\mapsto g(y)$ on $\sC\cap \sC^0$; then 
\begin{align*}
\sup_{y \in \sC\cap \sC^0} g(y) 
= \lim\limits_{n \to \infty} \sup_{y \in \sC\cap \sC^0_n}  g(y)
= \lim\limits_{n \to \infty}  g(\hat{y}_n)\quad\mbox{for}\quad\hat{y}_n:= (1-\tfrac{1}{n})\hat{y}.
\end{align*}
\end{lemma}
\begin{proof}
	As $\sC$ is convex and contains the origin, $\hat{y}^{\theta}_n \in \sC \cap \sC^0_n$. Moreover,
	\begin{equation*}
	\sup_{y \in \sC\cap \sC^0}g^{\theta}(y) 
	\geq \lim\limits_{n \to \infty} \sup_{y \in \sC\cap \sC^0_n}g^{\theta}(y) 
	\geq \limsup_{n \to \infty} g^{\theta} (\hat{y}^{\theta}_n )
	\end{equation*}
	as $\sC^0_{n} \subseteq \sC^0_{n+1} \subseteq \sC^0$.
	For the converse inequality, note that
	$g^{\theta}$ is concave and $g^{\theta}(0)=0$, so that
	\begin{equation*}
	g^{\theta}(\hat{y}^{\theta}_n) = g^{\theta}\big((1-\tfrac{1}{n})\, \hat{y}^{\theta}\big) \geq
	 (1-\tfrac{1}{n})\, g^{\theta}(\hat{y}^{\theta}).
	\end{equation*}
	Thus, we conclude that
	\begin{equation*}
	\liminf_{n \to \infty} g^{\theta}(\hat{y}^{\theta}_n)
	\geq g^{\theta}(\hat{y}^{\theta})
	=  \sup_{y \in \sC \cap \sC^0} g^{\theta}(y)
	\end{equation*}
	and the first claim follows. The proof of the second claim is analogous.
\end{proof}

\begin{lemma}\label{le:lsc}
The map $\theta \mapsto \sup_{y \in \sC \cap \sC^0} g^{\theta}(y)$ is real-valued and lower semicontinuous on $\Theta$.\end{lemma}
\begin{proof}
  We note that $\sup_{y \in \sC \cap \sC^0} g^{\theta}(y)>-\infty$ as $0 \in \sC \cap \sC^0$. On the other hand, the two conditions in Assumption~\ref{assumpt} yield that $\sup_{y \in \sC \cap \sC^0} g^{\theta}(y)<\infty$.
  
  We turn to the semicontinuity. Without loss of generality, we may assume that the truncation function $h$ is continuous; cf.\ \cite[Proposition~2.24, p.\,81]{JacodShiryaev.03}.
  Using the form of $g^{\theta}$ and the compactness of $\sC \cap \sC_0$, it suffices to show that for fixed $(b,c) \in \R^d \times \S^d_+$, the map
\begin{equation*}
  F \mapsto \fg(F) :=\sup_{y \in \sC \cap \sC^0} g^{(b,c,F)}(y)
\end{equation*}
is lower semicontinuous on $\cL_{\Theta}$. Consider the map
\begin{equation*}
F \mapsto \fg_{n}(F):=\sup_{y \in \sC \cap \sC_n^0} g^{(b,c,F)}(y)
\end{equation*}
for $n\in\N$. We deduce from Lemma~\ref{le:approx-g} that  $\fg_{n}(F)$ increases to $\fg(F)$ as $n\to\infty$. Therefore, it is sufficient to show that $\fg_{n}$ is lower semicontinuous for fixed $n$, and for this, in turn, it suffices to show that  $F \mapsto g^{(b,c,F)}(y)$ is lower semicontinuous on $\cL_{\Theta}$ for fixed  $y \in \sC \cap \sC^0_{n}$.

To see this, let
\begin{equation*}
I_{y}(z)=\left\{ \begin{array}{ll}
 \log(1+y^\top z)-y^\top h(z) \ &\mbox{if } \ p=0, \\
     &\\
        p^{-1}(1+y^\top z)^p -p^{-1}-y^\top h(z) \ &\mbox{if } \ p\neq 0
\end{array}\right.
\end{equation*}
denote the integrand in the definition of $g^{(b,c,F)}(y)$. Fix a continuous function $\psi_{n}:\R \to [0,1]$ which satisfies $\psi_{n}(u)=1$ for $u\geq -1+ \frac{1}{n}$ and $\psi_{n}(u)=0$ for $u<-1 +\frac{1}{2n}$. As $F[z \in \R^d \, | \, y^\top z< -1+\frac{1}{n}]=0$ and $\psi_n(y^\top z)=1$ on $\{z \in \R^d \, | \, y^\top z\geq -1+\frac{1}{n}\}$, we see that
\begin{equation*}
\int_{\R^d} I_{y}(z) \, F(dz)= \int_{\R^d} I_{y}(z)\, \psi_{n}(y^\top z) \, F(dz)= \int_{\R^d} I_{y,n}(z) \, F(dz), 
\end{equation*}
where we have set $I_{y,n}(z):= I_{y}(z)\, \psi_{n}(y^\top z)$.
Thus, it suffices to show that
\begin{equation*}
F \mapsto \int_{\R^d}I_{y,n}(z) \ F(dz)
\end{equation*}
is lower semicontinuous on $\cL_{\Theta}$. 
Let $F^k\to F$ be a convergent sequence in~$\cL_{\Theta}$.  As $h(z)=z$ in a neighborhood of $0$ and by the property of $\psi_n$, 
\begin{equation*}
z \mapsto  \frac{I_{y,n}(z)}{|z|^2 \wedge 1}
\end{equation*}
is continuous on $\R^d\setminus\{0\}$ and uniformly bounded from below by a constant~$K$. Define on $\R^d$ the function
\begin{equation*}
\tilde{I}_{y,n}(z) =\begin{cases}
    \frac{I_{y,n}(z)}{|z|^2 \wedge 1} & \mbox{if } \  z\neq 0,\\
   K& \mbox{if } \ z=0.
  \end{cases}
  \end{equation*}
  By construction, $z \mapsto \tilde{I}_{y,n}(z)$
is lower semicontinuous and uniformly bounded from below on $\R^d$. Thus, there exist bounded continuous functions $\tilde{I}^m_{y,n}(z)$ which increase to $\tilde{I}_{y,n}(z)$; cf.\ \cite[Lemma~7.14, p.\,147]{BertsekasShreve.78}. For any $F(dz)\in \cL$, let $\tilde{F}(dz):= |z|^2\wedge 1 . F(dz)$ be the finite measure 
$
A \mapsto \int_A |z|^2 \wedge 1 \, F(dz).
$
By the definition of the topology on $\mathcal{L}$, we have that $F^k\to F$  if and only if $\tilde{F}^k\to\tilde{F}$ in the sense of weak convergence. Using that $F(\{0\})=\tilde{F}(\{0\})=0$ for any $F\in \cL$ and Fatou's Lemma, we obtain that
\begin{align*}
\liminf_{k \to \infty} \int_{\R^d}I_{y,n}(z) \ F^k(dz) 
  = & \  \liminf_{k \to \infty} \int_{\R^d} \tilde{I}_{y,n}(z)  \ \tilde{F}^k(dz)\\
  \geq & \ \lim\limits_{m \to \infty} \, \liminf_{k \to \infty} \int_{\R^d} \tilde{I}^m_{y,n}(z)  \ \tilde{F}^k(dz)\\
  = & \ \lim\limits_{m \to \infty}  \int_{\R^d} \tilde{I}^m_{y,n}(z)  \ \tilde{F}(dz)\\
  \geq & \   \int_{\R^d} \tilde{I}_{y,n}(z)  \ \tilde{F}(dz)\\
 = & \ \int_{\R^d} I_{y,n}(z) \ F(dz).
\end{align*}
This completes the proof. 
\end{proof}

We can now show the relevant properties of the function $g^{\theta}(y)$ defined in~\eqref{eq:fct_g-bcF}.

\begin{proposition}\label{prop:local-minimax}
There exists $\hat{y} \in \sC\cap\sC^0$ such that
\begin{equation*}
\inf_{\theta\in \Theta} g^{\theta}(\hat{y})=\sup_{y \in \sC\cap\sC^0} \inf_{\theta\in \Theta} g^{\theta}(y)= \inf_{\theta\in \Theta} \sup_{y \in \sC\cap\sC^0}  g^{\theta}(y).
\end{equation*}
 If $\Theta\subseteq \R^d\times \S^d_+ \times \cL$ is compact,  there exists $\hat{\theta} \in \Theta$ such that
 $(\hat{\theta},\hat{y})$ is a saddle point for the function $g^{\theta}(y)$ on $\Theta\times \sC\cap\sC^0$.
\end{proposition}

\begin{proof}
Recall that $\sC\cap\sC^0$ and $\Theta$ are non-empty convex sets and that $\sC\cap\sC^0$ is compact. For fixed $\theta\in \Theta$, the function $y \mapsto g^{\theta}(y)$ is concave and upper semicontinuous (Lemma~\ref{le:fct_g}), whereas for fixed $y \in \sC\cap\sC^0$, the function $\theta \mapsto g^{\theta}(y)$ is convex. Thus, we deduce from Sion's minimax theorem~\cite[Theorem~4.2]{Sion.58} that 
\begin{equation*}
\inf_{\theta\in \Theta} \sup_{y \in \sC\cap\sC^0}  g^{\theta}(y) = \sup_{y \in \sC\cap\sC^0} \inf_{\theta\in \Theta} g^{\theta}(y).
\end{equation*}
To be precise, we require an extension of that theorem to functions taking values in $[-\infty,\infty)$; see, e.g., \cite[Appendix~E.2]{Pollard.13}.
As $y \mapsto \inf_{\theta\in \Theta} g^{\theta}(y)$ is upper semicontinuous (Lemma~\ref{le:fct_g}), we also obtain $\hat{y} \in \sC\cap\sC^0$ such that
\begin{equation*}
\inf_{\theta\in \Theta} g^{\theta}(\hat{y})=\sup_{y \in \sC\cap\sC^0} \inf_{\theta\in \Theta} g^{\theta}(y).
\end{equation*}
Assume that $\Theta$ is compact. Then, as  $\theta\mapsto \sup_{y \in \sC\cap\sC^0} g^{\theta}(y)$ is lower semicontinuous (Lemma~\ref{le:lsc}), there exists $\hat{\theta}\in \Theta$ such that
\begin{equation*}
 \sup_{y \in \sC\cap\sC^0}  g^{\hat{\theta}}(y)=\inf_{\theta\in \Theta} \sup_{y \in \sC\cap\sC^0}  g^{\theta}(y).
\end{equation*}
In view of the above minimax identity, $(\hat{\theta}, \hat{y})$ is a saddle point.
\end{proof}

\begin{remark}\label{rem:local-minimax}
For later use, we note that Proposition~\ref{prop:local-minimax} also holds true with respect to $\sC\cap \sC^0_n$ instead of $\sC\cap\sC^0$. Indeed, we may apply the proposition to the modified constraint $\tilde{\sC}=\sC\cap \sC^0_n$ instead of $\sC$.
\end{remark}

\section{Proofs for Logarithmic Utility}\label{se:proof-log}

In this section we focus on the logarithmic case $p=0$ and prove Theorems~\ref{thm} and~\ref{thm-compact}.  By scaling, we may assume that the initial capital is $x_{0}=1$, and we recall that Assumption~\ref{assumpt} is in force. Because the logarithmic utility turns out to be myopic under our specific setting of model uncertainty, the passage from the local results in the preceding section to the global ones is relatively direct.

\begin{lemma}\label{le:CanDec}
Let $P \in \fP$ have differential characteristics $\theta^{P}=(b^P,c^P,F^P)$ and let $\pi \in \cA$. Then
\begin{equation*}
E^P[\log(W^\pi_T)]= E^P\Big[\int_{0}^T g^{\theta^{P}_{s}}(\pi_s) \,ds\Big] \in [-\infty,\infty).
\end{equation*}
\end{lemma}
\begin{proof}
Let $\mu^X$ be the integer-valued random measure associated with the jumps of $X$. Under $P$, the stochastic integral $\int \pi\, dX$ has the canonical representation 
\begin{equation*}
\int_0^\cdot \pi_s \, dX_s = M^c + M^d + \int_0^\cdot \pi_s^\top b^P_s \,ds + \int_0^\cdot \pi_s \,dJ_s
\end{equation*}
where $M^{c}$ and $M^{d}$ are continuous and purely discontinuous local martingales, respectively, and
\begin{align*}
\langle M^c\rangle &= \int_0^\cdot \pi_s^\top c^P_s \pi_s \,ds,\\
[M^d] &=\int_0^\cdot \int_{\R^d} \big(\pi_s^\top h(z)\big)^2 \,\mu^X(dz,ds) ,\\
J &= \int_0^\cdot \int_{\R^d} \big(z-h(z)\big)\,\mu^X(dz,ds);
\end{align*}
 cf.\ \cite[Theorem~2.34, p.\,84]{JacodShiryaev.03} and \cite[Proposition~2.2]{NeufeldNutz.13a}. We claim that~$M^c$ and~$M^d$ are true martingales.
Indeed, Jensen's inequality and Assumption~\ref{assumpt} imply that 
\begin{align*}
  E^P\Big[ \big|[M^{d}]_T \big|^{1/2}\Big]
  & \leq E^P\bigg[\int_0^T \int_{\R^d} |\pi_s|^2\,|h(z)|^2   \, \mu^X(dz,ds) \bigg]^{1/2} \\
  & \leq C\, E^P\bigg[\int_0^T \int_{\R^d} |z|^2\wedge 1  \, F^P_{s}(dz)\,ds \bigg]^{1/2} \\
  & \leq  C\, \mathcal{K}^{1/2} T^{1/2}.
  \end{align*}
  for a constant $C$. Thus, the Burkholder--Davis--Gundy inequalities yield 
  \begin{equation*}\label{eq:estimateJumpMartBDG}
  E^P\bigg[\sup_{0\leq u \leq T} \big| M^d_u \big| \bigg] \leq C\, E^P\Big[\big|[M^d]_T \big|^{1/2}\Big] \leq C_\cK \,T^{1/2}.
  \end{equation*}
  Similarly, Assumption~\ref{assumpt} also implies that
  \begin{equation*}\label{eq:estimateContMartBDG}
    E^P\bigg[\sup_{0\leq u \leq T} \big| M^{c}_u \big| \bigg]\leq C_\cK\, T^{1/2}
\end{equation*}
and we conclude that $M^c$ and $M^d$ are true martingales. Recall that
$W^\pi>0$ and $W_{-}^\pi>0$ $P$-a.s.\ because $\pi\in\cA$; cf.\ \cite[Theorem~I.4.61, p.\,59]{JacodShiryaev.03}.  Thus, It\^o's formula yields that
\begin{align*}\log(W^\pi_T) &=  M^c_T +  M^d_T + \int_0^T \pi_s^\top b^P_s\,ds - \frac{1}{2}\int_0^T \pi_s^\top c^P_s \pi_s \,ds\\
& \quad\;+\int_0^T\int_{\R^d} \big[\log(1+\pi_s^\top z) -\pi_s^\top h(z)\big]\,\mu^X(dz,ds).
\end{align*}
In view of Assumption~\ref{assumpt} and \cite[Theorem~II.1.8, p.\,66]{JacodShiryaev.03}, taking expected values yields the result.
\end{proof}

The next three lemmas constitute the proof of Theorem~\ref{thm}. 

\begin{lemma}\label{le:C^0*}
Let $\hat{y}\in \argmax_{y \in \sC\cap\sC^0}\inf_{\theta\in \Theta} g^{\theta}(y)$. Then $\hat{y}$, seen as a constant process, is an element of $\cA$.
\end{lemma}

\begin{proof}
  We need to show that $W^{\hat{y}}>0$ $P$-a.s.\ for all $P \in \fP$, which by \cite[Theorem~I.4.61, p.\,59]{JacodShiryaev.03} is equivalent to $\hat{y} \in \sC^{0,*}$,
  where
  \begin{equation*}%
  \sC^{0,*} := \bigcap_{F \in \cL_{\Theta}} \big\{ y \in \R^d \, \big| \, F[z \in \R^d \, | \, y^\top z\leq -1]=0 \big\}.
  \end{equation*}
 As $0 \in \sC\cap\sC^0$ and $g^{\theta}(0)=0$ for all $\theta\in \Theta$, we have 
 \begin{equation*}
 \inf_{\theta\in \Theta} g^{\theta}(\hat{y})
 \geq \inf_{\theta\in \Theta} g^{\theta}(0)=0
 \end{equation*}
 and in particular $g^{\theta}(\hat{y})>-\infty$ for all $\theta\in \Theta$. The claim now follows from the definition of $g^{\theta}$.
\end{proof}

\begin{lemma}\label{le:proof-thm-l1-log}
	We have $u(1)<\infty$. Any 
	$
	  \hat{y}\in \argmax _{y \in \sC\cap \sC^0} \inf_{\theta\in \Theta} g^{\theta}(y)
	$
	satisfies
	\begin{equation*} %
	\inf_{P\in \fP}E^P[\log(W_T^{\hat{y}})]
	=\sup_{\pi \in \cA}\inf_{P\in \fP}E^P[\log(W_T^{\pi})]
	=\inf_{P\in \fP} \sup_{\pi \in \cA}E^P[\log(W_T^{\pi})]
	\end{equation*}
	and this value is given by $T \,\inf_{\theta\in \Theta} g^{\theta}(\hat{y})$.
\end{lemma}

\begin{proof}
	Let $\pi\in \cA$ and let $\theta^{P}$ denote the differential characteristics of $P$. Using Lemma~\ref{le:CanDec}, we have that
	\begin{align}
	\inf_{P\in \fP} E^P[\log(W^{\pi}_T)] &= \inf_{P\in \fP} E^P\Big[\int_0^T g^{\theta^P_s}(\pi_s)\,ds\Big] \nonumber\\
	&\leq \inf_{P\in \fP_{L}} E^P\Big[\int_0^T g^{\theta^P}(\pi_s)\,ds\Big], \nonumber\\
	&\leq \inf_{P\in \fP_{L}} E^P\Big[\int_0^T \sup_{y \in \sC\cap\sC^0} g^{\theta^P}(y)\,ds\Big] \nonumber \\
	&= T \inf_{\theta\in \Theta} \sup_{y\in \sC\cap \sC^0} g^{\theta}(y). \label{eq:pf-part1-2}
	\end{align}
	By Proposition~\ref{prop:local-minimax}, we have
	$
	\inf_{\theta\in \Theta} \sup_{y \in \sC\cap\sC^0}  g^{\theta}(y)
	= \inf_{\theta\in \Theta} g^{\theta}(\hat{y})
	$
	and thus
	\begin{align}
	\inf_{P\in \fP} E^P[\log(W^{\pi}_T)]
	&\leq T \inf_{\theta\in \Theta} g^{\theta}(\hat{y}) \nonumber\\
	&= \inf_{P\in \fP} E^P\Big[\int_0^T \inf_{\theta\in \Theta} g^{\theta}(\hat{y}) \, ds \Big] \nonumber\\
	&\leq  \inf_{P\in \fP} E^P\Big[\int_0^T  g^{\theta^P_s}(\hat{y}) \, ds \Big]. \nonumber\\
	&= \inf_{P\in \fP} E^P[\log(W^{\hat{y}}_T)],\nonumber
	\end{align}
	where  Lemma~\ref{le:CanDec} was again used. As $\pi \in \cA$ was arbitrary, we conclude that
	\begin{equation*}
	\sup_{\pi \in \cA} \inf_{P\in \fP} E^P[\log(W^{\pi}_T)]
	\leq T \inf_{\theta\in \Theta} g^{\theta}(\hat{y})
	\leq \inf_{P\in \fP} E^P[\log(W^{\hat{y}}_T)]. 
	\end{equation*}
  But then these inequalities must be equalities, as claimed.
	In particular, we have that $u(1)=T  \inf_{\theta\in \Theta} g^{\theta}(\hat{y})<\infty$; cf.\ Lemma~\ref{le:fct_g}.
	
	It remains to prove the minimax identity. To this end, note that
	\begin{align*}
	T  \inf_{\theta\in \Theta} \sup_{y \in \sC\cap\sC^0} g^{\theta}(y)
	&= \inf_{P \in \fP_{L}} E^P\Big[\int_0^T \sup_{y \in \sC\cap\sC^0} g^{\theta^P}(y) \, ds\Big] \nonumber\\
	&\geq \inf_{P \in \fP_{L}} \sup_{\pi \in \cA}E^P\Big[\int_0^T  g^{\theta^P}(\pi_s) \, ds\Big]\\
	&\geq \inf_{P \in \fP} \sup_{\pi \in \cA}E^P\Big[\int_0^T  g^{\theta_s^P}(\pi_s) \, ds\Big]. 
	\end{align*}
	Using also Proposition~\ref{prop:local-minimax} 
	and Lemma~\ref{le:CanDec}, we conclude  that
	\begin{align*}
	\sup_{\pi \in \cA} \inf_{P\in \fP} E^P[\log(W^{\pi}_T)] &=T  \inf_{\theta\in \Theta} \sup_{y \in \sC\cap\sC^0} g^{\theta}(y)\\
	&\geq \inf_{P \in \fP} \sup_{\pi \in \cA}E^P\Big[\int_0^T  g^{\theta_s^P}(\pi_s) \, ds\Big]\\
	&= \inf_{P \in \fP} \sup_{\pi \in \cA} E^P[\log(W^\pi_T)].
	\end{align*}
	The converse inequality is trivial, so the proof is complete.
\end{proof}

It remains to prove the third assertion of Theorem~\ref{thm}. 

\begin{lemma}\label{le:proof-thm-l2-log}
	Any constant $\tilde{\pi}\in \cA$ satisfying 
	\begin{equation} \label{eq:le:proof-thm-l2-log}
	\inf_{P\in \fP} E^P[\log(W^{\tilde{\pi}}_T)]=\sup_{\pi \in \cA} \inf_{P\in \fP} E^P[\log(W^{\pi}_T)]
	\end{equation}
	is an element of $\argmax _{y \in \sC\cap \sC^0} \inf_{\theta\in \Theta} g^{\theta}(y)$.
\end{lemma}

\begin{proof}
	We deduce from Lemma~\ref{le:proof-thm-l1-log}, \eqref{eq:pf-part1-2} and Proposition~\ref{prop:local-minimax} that
	\begin{align*}
	\inf_{P\in \fP} E^P[\log(W^{\tilde{\pi}}_T)] 
	&\leq \inf_{P\in \fP_{L}} E^P[\log(W^{\tilde{\pi}}_T)]\\
	& \leq   T  \sup_{y\in \sC\cap \sC^0} \inf_{\theta\in \Theta} g^{\theta}(y)\\
	& =\inf_{P\in \fP} E^P[\log(W^{\tilde{\pi}}_T)].
	\end{align*}
	Thus, the above inequalities are in fact equalities; in particular,
	\begin{equation*}
	 \inf_{P\in \fP_{L}} E^P[\log(W^{\tilde{\pi}}_T)] = T \sup_{y\in \sC\cap \sC^0} \inf_{\theta\in \Theta}g^{\theta}(y).
	\end{equation*}
	On the other hand, using Lemma~\ref{le:CanDec} and the fact that $\tilde{\pi}\in \cA$ is constant,
	\[
	\inf_{P\in \fP_{L}} E^P[\log(W^{\tilde{\pi}}_T)]
	= \inf_{P\in \fP_{L}}E^P\Big[\int_0^T g^{\theta^P}(\tilde{\pi})\,ds\Big]
	= T \inf_{\theta\in \Theta} g^{\theta}(\tilde{\pi}),
	\]
	so it follows that $\tilde{\pi}\in \argmax_{\sC\cap\sC^0}\inf_{\theta\in \Theta} g^{\theta}(y)$.
\end{proof}

The remaining two lemmas constitute the proof of the saddle point result, Theorem~\ref{thm-compact}.

\begin{lemma}\label{le:proof-thm-l3-log}
	Assume that $\Theta$ is compact. Let $\big(\hat{\theta},\hat{y}\big) \in \Theta \times\sC\cap\sC^0$  be a saddle point of the function $g^{\theta}(y)$ and let $\hat{P} \in \fP_{L}$ be the L\'evy law with triplet $\hat{\theta}$. Then
	$(\hat{P},\hat{y}) \in \fP_{L}\times \cA$ is a saddle point of $(P,\pi)\mapsto E^P[\log(W^\pi_T)]$ on $\fP \times \cA$ and
	\begin{equation*}
	\sup_{\pi\in \cA} \inf_{P\in\fP} E^P[\log(W^\pi_T)]
	= T \, g^{\hat{\theta}}(\hat{y}).
	\end{equation*}
\end{lemma}

\begin{proof}
	We recall that $\hat{y}\in\cA$; cf.\ Lemma~\ref{le:C^0*}. Let $\pi \in \cA$; then Lemma~\ref{le:CanDec} yields that
	\begin{equation*}
	E^{\hat{P}}[\log(W_T^{\pi})]= E^{\hat{P}}\Big[\int_0^T g^{\hat{\theta}}(\pi_s)\, ds\Big] \leq T \sup_{y \in \sC \cap \sC^0} g^{\hat{\theta}}(y).
	\end{equation*}
	Using the same lemma again,
	$
	 T \sup_{y \in \sC \cap \sC^0} g^{\hat{\theta}}(y)= T g^{\hat{\theta}}(\hat{y})= E^{\hat{P}}[\log(W_T^{\hat{y}})],
	$
	and as $\pi \in \cA$ was arbitrary,  we deduce that
	\begin{equation}\label{eq:pf1}
	\inf_{P \in \fP} \sup_{\pi \in \cA} E^{P}[\log(W_T^{\pi})]\leq \sup_{\pi \in \cA}E^{\hat{P}}[\log(W_T^{\pi})] \leq E^{\hat{P}}[\log(W_T^{\hat{y}})].
	\end{equation}
	As $(\hat{\theta},\hat{y})$ is a saddle point of the function $g^{\theta}(y)$,
	\begin{equation*}
	E^{\hat{P}}[\log(W_T^{\hat{y}})]= T\, g^{\hat{\theta}}(\hat{y}) = T \inf_{\theta\in \Theta} g^{\theta}(\hat{y}).
	\end{equation*}
	Moreover, using Lemma~\ref{le:CanDec}, we have
	\begin{align*}
	T \inf_{\theta\in \Theta} g^{\theta}(\hat{y}) 
	&= \inf_{P \in \fP}E^P\Big[\int_0^T \inf_{\theta\in \Theta} g^{\theta}(\hat{y}) \, ds\Big]\\
	&\leq \inf_{P \in \fP}E^P\Big[\int_0^T g^{\theta^P_s}(\hat{y}) \, ds\Big]\\
	&=\inf_{P \in \fP} E^{P}[\log(W_T^{\hat{y}})].
	\end{align*}
	Thus, we obtain that
	\begin{equation}\label{eq:pf2}
	E^{\hat{P}}[\log(W_T^{\hat{y}})]\leq \inf_{P \in \fP} E^{P}[\log(W_T^{\hat{y}})]\leq \sup_{\pi \in \cA} \inf_{P \in \fP} E^{P}[\log(W_T^{\pi})].
	\end{equation}
	It follows that the inequalities in~\eqref{eq:pf1} and~\eqref{eq:pf2} are in fact equalities, and the proof is complete.
\end{proof}

\begin{lemma}\label{le:proof-thm-l4-log}
	Assume that $\Theta$ is compact and let $(\tilde{P},\tilde{\pi})$ be a saddle point of $(P,\pi)\mapsto E^P[\log(W^\pi)]$ on $\fP\times\cA$
	with $\tilde{P}\in \fP_{L}$ and $\tilde{\pi}$ constant. Then $(\tilde{\theta},\tilde{\pi})$ is a saddle point of the function $g^{\theta}(y)$ on $\Theta \times\sC\cap\sC^0$, where $\tilde{\theta}$ is the L\'evy triplet of $\tilde{P}$.
\end{lemma}

\begin{proof}
	As $\tilde{P} \in \fP_{L}$, we have
	\begin{equation*}
	E^{\tilde{P}}[\log(W^{\tilde{\pi}}_T)]
	=\inf_{P \in \fP}E^{P}[\log(W^{\tilde{\pi}}_T)]
	\leq \inf_{P \in \fP_{L}}E^{P}[\log(W^{\tilde{\pi}}_T)]
	\leq E^{\tilde{P}}[\log(W^{\tilde{\pi}}_T)],
	\end{equation*}
	which implies that the above inequalities are equalities. By Lemma~\ref{le:CanDec}, as $\tilde{\pi}$ is constant, we obtain that
	\begin{equation}\label{eq:pf-part4-1}
	T\, g^{\tilde{\theta}}(\tilde{\pi})=E^{\tilde{P}}[\log(W^{\tilde{\pi}}_T)]
	=\inf_{P \in \fP_{L}}E^{P}[\log(W^{\tilde{\pi}}_T)]
	=T \inf_{\theta\in \Theta} g^{\theta}(\tilde{\pi}).
	\end{equation}
	Recall from the proof of Lemma~\ref{le:C^0*} that a constant process $y$ is in $\cA$ if and only if $y \in \sC\cap\sC^{0,*}$. As $(\tilde{P},\tilde{\pi})$ is a saddle point, we deduce  that
	\begin{align*}
	E^{\tilde{P}}[\log(W^{\tilde{\pi}}_T)]
	= \sup_{\pi \in \cA} E^{\tilde{P}}[\log(W^{\pi}_T)]
	& = \,\sup_{\pi \in \cA} E^{\tilde{P}}\Big[\int_0^T g^{\tilde{\theta}}(\pi_s) \, ds \Big]\\
	& \geq \,\sup_{y \in \sC\cap \sC^{0,*}} E^{\tilde{P}}\Big[\int_0^T g^{\tilde{\theta}}(y) \, ds \Big]\\ 
	& \geq\, E^{\tilde{P}}\Big[\int_0^T g^{\tilde{\theta}}(\tilde{\pi}) \, ds \Big]\\
	& =\, E^{\tilde{P}}[\log(W^{\tilde{\pi}}_T)].
	\end{align*} 
	Therefore, again, the above inequalities are in fact equalities. In particular,
	\[
	T\, g^{\tilde{\theta}}(\tilde{\pi})
	=E^{\tilde{P}}[\log(W^{\tilde{\pi}}_T)]
	= \sup_{y \in \sC\cap \sC^{0,*}} E^{\tilde{P}}\Big[\int_0^T g^{\tilde{\theta}}(y) \, ds \Big] 
	= T \sup_{y \in \sC\cap \sC^{0,*}} g^{\tilde{\theta}}(y),
	\]
	and in the last expression we may replace $\sC \cap \sC^{0,*}$ by its closure $\sC\cap \sC^0$ as 
	$g^{\tilde{\theta}}$ is concave and proper. Together with \eqref{eq:pf-part4-1}, this shows that $(\tilde{\theta},\tilde{\pi})$ is a saddle point of the function $g^{\theta}(y)$.
\end{proof}

\section{Proofs for Power Utility}\label{se:proof-pow}

In this section, we focus on the case $U(x)=\frac{1}{p}x^p$, where $p\in (-\infty,0)\cup (0,1)$. We recall that Assumption~\ref{assumpt} is in force and that the initial capital is $x_{0}=1$, without loss of generality.

The arguments for power utility are less direct than in the logarithmic case, because the power utility investor is typically not myopic. Thus, the optimal strategy and expected utility for a fixed $P$ cannot be expressed by the corresponding function $g^{\theta}$ (see \cite{Nutz.09a, Nutz.09b} for the structure in the general case). However, the power utility problem remains tractable when $P$ is a L\'evy law, and we shall see that the worst case over all L\'evy laws $P\in\fP_{L}$ already corresponds to the worst case over all $P\in\fP$. The crucial tool to prove that is a martingale argument, contained in Lemma~\ref{le:expo-martingale} below.

For some of the arguments we need to avoid the singularity of $U$ at zero and the corresponding singularity of $g$ at the boundary of $\sC^{0}$. To this end, recall that
\begin{equation*}
\sC^0_{n} = \bigcap_{F \in \cL_{\Theta}} \Big\{ y \in \R^d \, \Big| \, F\big[z \in \R^d \, \big| \, y^\top z< -1+\tfrac{1}{n}\big]=0 \Big\} \subseteq  \sC^{0}
\end{equation*}
for $n \in \N$ and define $\cA_n$ as the set of all predictable processes $\pi$ such that $\pi_t(\omega) \in \sC\cap\sC^0_n$ for all $(\omega,t)\in [\![0,T]\!]$. This implies that $W^{\pi}>0$ $P$-a.s.\ for all $P\in\fP$ and in particular $\pi\in\cA$.

\begin{lemma}\label{le:expo-martingale}
  Let $P \in \fP$ and let $\theta^P=(b^{P},c^{P},F^{P})$ be the corresponding differential characteristics. If $\pi \in \cA_n$ for some  $n \in \N$, then
  \[
  M_t:=\frac{ (W_t^\pi)^p}{\exp\big( p \int_0^t g^{\theta_s^P}(\pi_s)\,ds\big)}
  \]
  is real-valued and $M=(M_t)_{t\leq T}$ is a martingale with unit expectation. 
  
  If
   $P \in \fP_{L}$ and $\pi\in \sC\cap\sC^{0}$ is constant, then
  \begin{equation}\label{eq:powerLevyExp}
  E^P[U(W_T^{\pi})]= \frac{1}{p}\exp\Big(p\,T\, g^{\theta^{P}}(\pi)\Big) \in [-\infty,\infty).
  \end{equation}
\end{lemma}

	\begin{proof}
	  Let $\pi \in \cA_n$; then the function $g^{\theta_s^P}(\pi_s)$ and its integral are finite. Moreover, both $W^\pi$ and $W^\pi_{-}$ are strictly positive; cf.\ \cite[Theorem~I.4.61, p.\,59]{JacodShiryaev.03}. Thus, the process $M$ is a semimartingale with values in $(0,\infty)$. In particular, its drift rate 
	\[
	  a^{M} := b^{M} + \int (z-h(z))\,F^{M}(dz)
	\]
 is well-defined with values in $(-\infty,\infty]$; cf.\ \cite[Remark~2.3]{Nutz.09b}. Moreover, $M$ is a $\sigma$-martingale and true supermartingale as soon as $a^{M}=0$; see, e.g., \cite[Lemma~2.4]{Nutz.09b}.
	  
	Set $(b,c,F)=(b^{P},c^{P},F^{P})$ and $Y=(W^\pi)^p$. An application of It\^o's formula shows that the drift rate $a^{Y}$ of $Y$ satisfies
	\begin{align*}
	        \frac{a^{Y}}{Y_{-}} & = p\pi^\top b + \frac{p(p-1)}{2} \pi^\top c \pi + \int \big[(1+\pi^\top z)^p - 1 - p\pi^\top h(z)\big]\, \,F(dz)\\
	        &= pg^{\theta}(\pi).
	\end{align*}
	See, e.g., \cite[Lemma~3.4]{Nutz.09b} for a similar calculation.
	Noting that the process $G_{t}=\exp\big( p \int_0^t g^{\theta_s}(\pi_s)\,ds\big)$ is continuous and of finite variation, we have
	\[
	  dM = G^{-1}\,dY + Y_{-} \,d(G^{-1}) = G^{-1}\,dY - Y_{-}G^{-1} pg^{\theta}(\pi) \,dt
	\]
	and it follows that $a^{M}= G^{-1} a^{Y}- Y_{-}G^{-1} pg^{\theta}(\pi)=0$. As a result, $M$ is a $\sigma$-martingale and a supermartingale.
	To establish that $M$ is a true martingale, it remains to show that $M$ is of class (D). 
	
	Consider first the case $p\in (0,1)$. Let $\varepsilon>0$ be as in Assumption~\ref{assumpt} and let $\tau\leq T$ be a stopping time; we estimate
	\begin{equation*}
	E\big[|M_\tau|^{1+\varepsilon}\big]=E\Bigg[\frac{\big( W_\tau^\pi\big)^{p(1+\varepsilon)}}{\exp\big( p(1+\varepsilon) \int_0^\tau g^{\theta_s}(\pi_s)\,ds\big)}\Bigg].
	\end{equation*}
	Set $\tilde{p}:=p(1+\varepsilon)$ and let $g^{\theta}(\tilde{p},\pi)$ be defined like $g^{\theta}(\pi)$ but with $p$ replaced by $\tilde{p}$. Using the supermartingale property of $M$ with respect to $\tilde{p}$ (which holds by the same arguments), we obtain that 
	\begin{align*}
	& \ E\bigg[\frac{\big( W_\tau^\pi\big)^{p(1+\varepsilon)}}{\exp\big( p(1+\varepsilon) \int_0^\tau g^{\theta_s}(\pi_s)\,ds\big)}\bigg]\\
	= & \ E\bigg[\frac{\big( W_\tau^\pi\big)^{\tilde{p}}}{\exp\big(\tilde{p} \int_0^\tau g^{\theta_s}(\tilde{p},\pi_s)\,ds\big)} \, \frac{\exp\big(\tilde{p} \int_0^\tau g^{\theta_s}(\tilde{p},\pi_s)\,ds\big)}{\exp\big(\tilde{p} \int_0^\tau g^{\theta_s}(\pi_s)\,ds\big)}\bigg]\\
	\leq & \ \frac{\exp\Big(\tilde{p} TC_{1}\, \sup_{(b,c,F)\in \Theta}  \big\{ |b| + \frac{|\tilde{p}-1|}{2} |c| + \int |z|^2 \wedge |z|^{\tilde{p}}\,F(dz) \big\}\Big)}{\exp\Big(-\tilde{p} TC_{2}\, \sup_{(b,c,F)\in \Theta} \big\{  |b| + \frac{|p-1|}{2}   |c|  +  \int |z|^2 \wedge |z|^{p}\,F(dz) \big\}\Big)}%
	\end{align*}
	where $C_1,C_2$ are finite constants depending only on $p,\tilde{p}$ and the diameter of $\sC\cap\sC^0$. The last line is finite due to Assumption~\ref{assumpt} and does not depend on $\tau$. We have shown that
	\begin{equation*}
	  \sup_{\tau\leq T} E\big[ |M_\tau|^{1+\varepsilon} \big]<\infty,
	\end{equation*}
	so the de la Vall\'ee-Poussin theorem implies that $M$ is of class (D) and in particular a true martingale.
	
	For the case $p<0$, choose an arbitrary $\eps>0$ and recall that $\pi \in \cA_n$. A similar estimate as above holds for $\tilde{p}=p(1+\varepsilon)<0$, except that the the signs are reversed, the constants $C_{1}, C_{2}$ now depend on the fixed $n$,  and $|z|^2 \wedge |z|^{p}$, $|z|^2 \wedge |z|^{\tilde p}$ are replaced by $|z|^2 \wedge 1$. The conclusion remains the same.
	
	Finally, let $P \in \fP_{L}$ and let $\pi\in \sC\cap\sC^{0}$ be constant. We observe that $g^{\theta}(\pi)\in [-\infty,\infty)$, and the value $-\infty$ can occur only if $p<0$. If $g^{\theta}(\pi)$ is finite and $W^{\pi}>0$ $P$-a.s., the above arguments still apply and the identity~\eqref{eq:powerLevyExp} follows.
	
	Let $p\in(0,1)$. We have just seen that~\eqref{eq:powerLevyExp} holds when $\pi\in\cA_{n}$, and then the general case follows by passing to the limit on both sides; cf.\ Lemma~\ref{le:stoch-expo-limit} below.
	
	Let $p<0$. If $P[W^{\pi}_{T}=0]>0$, then $F[z \in \R^d \,| \, y^\top z= -1]>0$ and both sides of~\eqref{eq:powerLevyExp} equal $-\infty$. If $W^{\pi}>0$ $P$-a.s.\ but $g^{\theta}(\pi)=-\infty$, we need to argue that $E^P[Y_T]=\infty$ for $Y=(W^{\pi})^{p}$. Suppose that $E^P[Y_T]<\infty$. Then as $Y$ is the exponential of a L{\'e}vy process \cite[Lemma~4.2]{Kallsen.00}, $Y$ is of class~(D) on $[0,T]$ and a special semimartingale \cite[Lemma~4.4]{Kallsen.00}. In particular, its drift rate $a^{Y}=Y_{-}p g^{\theta}(\pi)$ has to be finite \cite[Proposition~II.2.29, p.\,82]{JacodShiryaev.03}. This contradicts $g^{\theta}(\pi)=-\infty$ and completes the proof.
\end{proof}

In the next two lemmas, we prove our main results for the set $\cA_{n}$ of strategies, where $n$ is fixed. We shall pass to the desired set~$\cA$ in a later step.

\begin{lemma}\label{le:proof-thm-l1-pow}
	Let $\hat{y}\in  \argmax _{y \in \sC\cap \sC^0_n} \inf_{\theta\in \Theta} g^{\theta}(y)$; then 
	\begin{equation*} %
	\inf_{P\in \fP}E^P[U(W_T^{\hat{y}})]
	=\sup_{\pi \in \cA_n}\inf_{P\in \fP}E^P[U(W_T^{\pi})]
	=\inf_{P\in \fP} \sup_{\pi \in \cA_n}E^P[U(W_T^{\pi})]
	\end{equation*}
	and this value is given by $\frac{1}{p}\exp\big(p\,T \,\inf_{\theta\in \Theta} g^{\theta}(\hat{y})\big)$.
\end{lemma}

\begin{proof}
	Let $\pi \in \cA_n$. The classical result for power utility maximization in the L\'evy setting, see \cite[Theorem~3.2]{Nutz.09c}, yields that
	\begin{equation*}
	\inf_{P \in \fP}E^P[U(W_T^{\pi})]
	\leq \inf_{P \in \fP_{L}}E^P[U(W_T^{\pi})]
	\leq
	\inf_{P \in \fP_{L}} \sup_{y \in \sC \cap \sC^0_n} \frac{1}{p}\exp\Big(p\,T \,g^{\theta^P}(y)\Big).
	 \end{equation*}
	In view of Proposition~\ref{prop:local-minimax} and the definition of $\hat{y}$,
	\begin{equation*}
	\inf_{P \in \fP_{L}} \sup_{y \in \sC \cap \sC^0_n} \frac{1}{p}\exp\Big(p\,T \,g^{\theta^P}(y)\Big)
	=
	\frac{1}{p}\exp\Big(p\,T \,\inf_{\theta\in \Theta} g^{\theta}(\hat{y})\Big).
	\end{equation*}
	Moreover, by Lemma~\ref{le:expo-martingale}, we have for any $P \in \fP$ that
	\begin{align*}
	  \frac{1}{p}\exp\Big(p\,T \,\inf_{\theta\in \Theta} g^{\theta}(\hat{y})\Big)
	&= E^P\Bigg[\frac{U(W_T^{\hat{y}})}{\frac{1}{p}\exp\big(p\, \int_0^T g^{\theta^P_s}(\hat{y})\,ds\big)}\Bigg] \frac{1}{p}\exp\Big(p\,T \,\inf_{\theta\in \Theta} g^{\theta}(\hat{y})\Big)\\
	 &\leq  E^P[U(W_T^{\hat{y}})].
	\end{align*}
	As $\pi \in \cA_n$ and $P \in \fP$ were arbitrary, we conclude that
	\begin{equation*}
	\sup_{\pi \in \cA_n}\inf_{P \in \fP} E^P[U(W_T^{\pi})]\leq \frac{1}{p}\exp\Big(p\,T \,\inf_{\theta\in \Theta} g^{\theta}(\hat{y})\Big) \leq  \inf_{P \in \fP} E^P[U(W_T^{\hat{y}})].
	\end{equation*}
	As $\hat{y} \in \cA_n$, these inequalities must be equalities.
	
	It remains to prove the minimax identity. By the definition of $\hat{y}$ and the classical result in \cite[Theorem~3.2]{Nutz.09c},
	\begin{align*}
	\frac{1}{p}\exp\Big(p\,T \,\inf_{\theta\in \Theta} g^{\theta}(\hat{y})\Big)
	= & \ \inf_{P \in \fP_{L}} \sup_{y \in \sC\cap \sC^0_n} \frac{1}{p}\exp\Big(p\,T \, g^{\theta^P}(y)\Big)\\
	=& \  \inf_{P \in \fP_{L}}\sup_{\pi\in \cA_n} E^P[U(W_T^{\pi})]\\
	\geq & \  \inf_{P \in \fP}\sup_{\pi\in \cA_n} E^P[U(W_T^{\pi})].
	\end{align*}
	Together with the above, we have
	\begin{equation*}
	\sup_{\pi \in \cA_n}\inf_{P \in \fP} E^P[U(W_T^{\pi})]= \frac{1}{p}\exp\Big(p\,T \,\inf_{\theta\in \Theta} g^{\theta}(\hat{y})\Big)\geq \inf_{P \in \fP}\sup_{\pi\in \cA_n} E^P[U(W_T^{\pi})],
	\end{equation*}
	and the converse inequality is clear.
\end{proof}

\begin{lemma}\label{le:proof-thm-l3-pow}
	Assume that $\Theta$ is compact. Let $(\hat{\theta},\hat{y})$  be a saddle point of the function $g^{\theta}(y)$ on $\Theta \times\sC\cap\sC^0_n$ and let $\hat{P} \in \fP_{L}$ be the L\'evy law with triplet $\hat{\theta}$. Then $(\hat{P},\hat{y}) \in \fP_{L}\times \cA_n$ is a saddle point of $(P,\pi)\mapsto E^P[U(W^\pi_T)]$ on $\fP \times \cA_n$ and
	\begin{equation*}
	\sup_{\pi\in \cA_n} \inf_{P\in\fP} E^{P}[U(W^{\pi}_T)]
	= \frac{1}{p}\exp\Big(p\,T \, g^{\hat{\theta}}(\hat{y})\Big).
	\end{equation*}
\end{lemma}

\begin{proof}
  The line of argument is the same as in Lemma~\ref{le:proof-thm-l3-log} for the logarithmic case, except that the use of Lemma~\ref{le:CanDec} needs to be substituted by Lemma~\ref{le:expo-martingale} and a martingale argument, much like in the preceding proof. We omit the details.
\end{proof}

\begin{remark}\label{rem:le-Levy-pow}
  For later reference, we record that Lemmas~\ref{le:proof-thm-l1-pow} and~\ref{le:proof-thm-l3-pow} remain true if $\fP$ is replaced by $\fP_{L}$ in the assertion. 
\end{remark}

Our next goal is to obtain the preceding two results for $\cA$ and $\sC^{0}$ rather than the auxiliary sets $\cA_{n}$ and $\sC^{0}_{n}$. This will be achieved by passing to the limit as $n\to\infty$, for which some preparations are necessary.

\begin{lemma}\label{le:stoch-expo-limit}
	Let $P \in \fP$ and $\pi \in \cA$. Then $\pi_n:=(1-\frac{1}{n}) \pi \in\cA_{n}$ and
	\begin{equation*}
	\limsup_{n \to \infty} E^P[U(W^{\pi_n}_T)]\leq E^P[U(W^{\pi}_T)].
	\end{equation*}
	Moreover, if $p\in (0,1)$, then $U(W^{\pi_n}_T)\to U(W^{\pi}_T)$ in $L^1(P)$.
\end{lemma}

\begin{proof}
	It is clear that $\pi_n\in \cA_n$.
	Using that $W_T^\pi= \cE(\int \pi dX)_T$, standard arguments show that $W_T^{\pi_n}$ converges %
	$P$-a.s.\ to  $W_T^\pi$, and then $U(W_T^{\pi_n})$ converges %
	 $P$-a.s.\ to $U(W_T^{\pi})$. When $p<0$, we have $U\leq0$ and the result follows from Fatou's Lemma. 
	 For $p\in (0,1)$, let $\varepsilon>0$ be as in Assumption~\ref{assumpt} and set $\tilde{p}:=p(1+\varepsilon)$. An estimate as in the proof of Lemma~\ref{le:expo-martingale} yields that
	 \begin{equation*}
	  E^P\big[|U(W^{\pi_n}_T)|^{1+\varepsilon}\big] 
	 \leq K<\infty
	 \end{equation*}	 
	 for all $n$, where $K$ is a constant depending on $\tilde{p}:=p(1+\varepsilon)$, the diameter of $\sC \cap \sC^0$ and $\cK$. Thus, $(U(W^{\pi_n}_T))_{n \in \N}$ is uniformly integrable and the convergence in $L^1(P)$ follows. 
\end{proof}

As we will be using results from the classical utility maximization problem \cite{Nutz.09c}, let us comment on a subtlety regarding the class of  strategies. Let $P \in \fP$ and denote by $\cA^P$ the set of all predictable  processes taking values in $\sC\cap \sC^0$ such that $W^\pi>0$ $P$-a.s.; this is the class of admissible strategies in \cite{Nutz.09c} if $\sC\cap \sC^0$ is used as the constraint set (which is necessarily contained in the natural constraints with respect to $P$). In the case $p>0$, we have $\cA \supseteq \cA^P$ as we did not enforce strict positivity in the definition of $\cA$. On the other hand, in the case $p<0$, we have required positivity under \emph{all} models in $\fP$, which results in an inclusion $\cA \subseteq \cA^P$. For the set $\sC\cap \sC^0_{n}$ that has been used above, no such subtleties exist as the wealth process is automatically strictly positive under all models.

\begin{lemma}\label{le:levyvalue}
 Let $P \in \fP_L$; then
 \begin{equation*}
 \sup_{\pi \in \cA^P}E^P[U(W^\pi_T)]\geq \sup_{\pi \in \cA}E^P[U(W^\pi_T)].
 \end{equation*}
 Moreover, if $p\in (0,1)$, we have equality.
\end{lemma}

\begin{proof}
 If $p<0$, the claim is clear as $\cA \subseteq \cA^P$. Let $p\in (0,1)$; then $\cA^{P} \subseteq \cA$, so it suffices to show the stated inequality. Let $\pi_n =(1-\frac{1}{n}) \pi \in \cA_n$ for $\pi \in \cA$. Lemma~\ref{le:stoch-expo-limit} yields that
 \begin{equation*}
 \sup_{\pi \in \cA}E^P[U(W^\pi_T)]=\sup_{\pi \in \cA} \lim\limits_{n \to \infty} E^P[U(W^{\pi_n}_T)]\leq  \lim\limits_{n \to \infty} \sup_{\pi \in \cA_n} E^P[U(W^{\pi}_T)].
 \end{equation*}
Let $\theta$ be the L\'evy triplet of $P$. We deduce from \cite[Theorem~3.2]{Nutz.09c} and Lemma~\ref{le:approx-g} that
 \begin{align*}
\lim\limits_{n \to \infty} \sup_{\pi \in \cA_n} E^P[U(W^{\pi}_T)]
&= \lim\limits_{n \to \infty} \frac{1}{p} \exp\Big(p\,T\, \sup_{y \in \sC \cap \sC^0_n}  g^{\theta}(y)\Big)\\
&= \frac{1}{p} \exp\Big(p\,T\, \sup_{y \in \sC \cap \sC^0}  g^{\theta}(y)\Big)\\
&=\sup_{\pi \in \cA^P} E^P[U(W^{\pi}_T)].
 \end{align*}
\end{proof}

We can now prove the main lemma for the passage from $\cA_n$ to $\cA$.

\begin{lemma}\label{le:proof-thm-limit-l1-pow}
 We have
 \begin{align*}
 \lim\limits_{n \to \infty} \sup_{\pi \in \cA_n}\inf_{P \in \fP} E^P[U(W_T^{\pi})]
 &= \frac{1}{p} \exp\Big(p\,T\, \sup_{y \in \sC \cap \sC^0} \inf_{\theta \in \Theta} g^{\theta}(y)\Big)\\
 & = \sup_{\pi \in \cA}\inf_{P \in \fP} E^P[U(W_T^{\pi})].
 \end{align*}
 In particular, $u(1)<\infty$.
\end{lemma}

\begin{proof}
 As $\cA_n\subseteq \cA_{n+1}\subseteq \cA$, the limit exists and
  \begin{equation*}
  \lim\limits_{n \to \infty} \sup_{\pi \in \cA_n}\inf_{P \in \fP} E^P[U(W_T^{\pi})]\leq \sup_{\pi \in \cA}\inf_{P \in \fP} E^P[U(W_T^{\pi})].
  \end{equation*}
 On the other hand, for each $n \in \N$, the minimax result of Lemma~\ref{le:proof-thm-l1-pow} and Remark~\ref{rem:le-Levy-pow} yield that
 \begin{align*}
   \sup_{\pi \in \cA_n}\inf_{P \in \fP} E^P[U(W_T^{\pi})]
  &=   \inf_{P \in \fP} \sup_{\pi \in \cA_n} E^P[U(W_T^{\pi})]
  =  \inf_{P \in \fP_{L}} \sup_{\pi \in \cA_n} E^P[U(W_T^{\pi})].
 \end{align*}
 Applying the classical result of~\cite[Theorem~3.2]{Nutz.09c} for each $P \in \fP_{L}$, we have \begin{equation*}
 \inf_{P \in \fP_{L}} \sup_{\pi \in \cA_n} E^P[U(W_T^{\pi})]
 =  \inf_{P \in \fP_{L}} \frac{1}{p} \exp\Big(p\,T \sup_{y \in \sC \cap \sC^0_n}g^{\theta^P}(y)\Big).
 \end{equation*}
  Using the local minimax result of Proposition~\ref{prop:local-minimax} with respect to $\sC\cap\sC^0_n$, cf.\  Remark~\ref{rem:local-minimax},
  \begin{align*}
  \inf_{P \in \fP_{L}} \frac{1}{p} \exp\Big(p\,T \sup_{y \in \sC \cap \sC^0_n}g^{\theta^P}(y)\Big)
   = \frac{1}{p} \exp\Big(p\,T \sup_{y \in \sC \cap \sC^0_n} \inf_{\theta \in \Theta} g^{\theta}(y)\Big).
  \end{align*}
  By Lemma~\ref{le:approx-g} and, once again, Proposition~\ref{prop:local-minimax},
  \begin{align*}
  \lim\limits_{n \to \infty}  \frac{1}{p} \exp\Big(p\,T \sup_{y \in \sC \cap \sC^0_n} \inf_{\theta \in \Theta} g^{\theta}(y)\Big)
  &=  \frac{1}{p} \exp\Big(p\,T \sup_{y \in \sC \cap \sC^0} \inf_{\theta \in \Theta} g^{\theta}(y)\Big)\\
  &=  \frac{1}{p} \exp\Big(p\,T \inf_{\theta \in \Theta} \sup_{y \in \sC \cap \sC^0} g^{\theta}(y)\Big).
  \end{align*}
  We deduce from \cite[Theorem~3.2]{Nutz.09c} %
  and Lemma~\ref{le:levyvalue} that
  \begin{align*}
  \frac{1}{p} \exp\Big(p\,T \inf_{\theta \in \Theta} \sup_{y \in \sC \cap \sC^0} g^{\theta}(y)\Big)
  = & \ \inf_{P \in \fP_{L}} \frac{1}{p} \exp\Big(p\,T \sup_{y \in \sC \cap \sC^0}g^{\theta^P}(y)\Big)\\
 = & \ \inf_{P \in \fP_{L}} \sup_{\pi \in \cA^P} E^P[U(W_T^{\pi})]\\
 \geq & \ \inf_{P \in \fP_{L}} \sup_{\pi \in \cA} E^P[U(W_T^{\pi})].
  \end{align*}
 Noting also the trivial inequalities
 \begin{equation*}
 \inf_{P \in \fP_{L}} \sup_{\pi \in \cA} E^P[U(W_T^{\pi})]\geq \inf_{P \in \fP} \sup_{\pi \in \cA} E^P[U(W_T^{\pi})]\geq \sup_{\pi \in \cA} \inf_{P \in \fP}  E^P[U(W_T^{\pi})],
 \end{equation*}
 we have established that
 \begin{align*}
 \sup_{\pi \in \cA} \inf_{P \in \fP}  E^P[U(W_T^{\pi})]
 & \geq \lim\limits_{n \to \infty} \sup_{\pi \in \cA_n}\inf_{P \in \fP} E^P[U(W_T^{\pi})]\\
 & = \frac{1}{p} \exp\Big(p\,T \sup_{y \in \sC \cap \sC^0} \inf_{\theta \in \Theta} g^{\theta}(y)\Big)\\
 & \geq \sup_{\pi \in \cA} \inf_{P \in \fP}  E^P[U(W_T^{\pi})]
 \end{align*}
 and hence all these expressions are equal.
\end{proof}

We are now ready to finish the proof of parts (i) and (ii) of Theorem~\ref{thm}.

\begin{lemma}\label{le:proof-thm-limit-l2-pow}
 Let $\hat{y} \in \argmax_{y \in \sC \cap \sC^0} \inf_{\theta\in \Theta} g^{\theta}(y)$; then
 \begin{equation*}
  \inf_{P \in \fP} E^P[U(W_T^{\hat{y}})]
  = \sup_{\pi \in \cA} \inf_{P \in \fP} E^P[U(W_T^{\pi})]
  = \inf_{P \in \fP} \sup_{\pi \in \cA}  E^P[U(W_T^{\pi})].
 \end{equation*}
\end{lemma}

\begin{proof}
  We first note that $\hat{y}\in \cA$. This is obvious from the definition of $\cA$ for $p>0$, whereas for $p<0$ the proof is identical to Lemma~\ref{le:C^0*}. As a result,
 \begin{equation*}
  \inf_{P \in \fP} E^P[U(W_T^{\hat{y}})]
  \leq \sup_{\pi \in \cA}\inf_{P \in \fP} E^P[U(W_T^{\pi})]
  \leq \inf_{P \in \fP} \sup_{\pi \in \cA}  E^P[U(W_T^{\pi})].
  \end{equation*}
  We first prove the converse to the first inequality. 
  By Lemma~\ref{le:proof-thm-limit-l1-pow}, it suffices to show that
  \begin{equation*}
  \inf_{P \in \fP} E^P[U(W_T^{\hat{y}})]\geq \frac{1}{p}\exp\Big(p\,T \sup_{y \in \sC \cap \sC^0}\inf_{\theta\in \Theta} g^{\theta}(y)\Big).
  \end{equation*}
  Indeed, Lemma~\ref{le:stoch-expo-limit} shows that $\hat{y}_n:=(1-\frac{1}{n})\hat{y}\in \cA_{n}$ satisfies
  \begin{equation*}
   \inf_{P \in \fP} E^P[U(W_T^{\hat{y}})]\geq 
   \inf_{P \in \fP} \limsup_{n \to \infty} E^P[U(W_T^{\hat{y}_n})],
  \end{equation*}
  while Lemma~\ref{le:expo-martingale} 
  yields
  \begin{align*}
  & \ \inf_{P \in \fP} \limsup_{n \to \infty} E^P[U(W_T^{\hat{y}_n})]\\
   = & \ \inf_{P \in \fP} \limsup_{n \to \infty} E^P\Bigg[\frac{U(W_T^{\hat{y}_n})}{\frac{1}{p}\exp\big(p \int_0^T g^{\theta^P_s}(\hat{y}_n)\,ds\big)} \ \frac{1}{p}\exp\Big(p \int_0^T g^{\theta^P_s}(\hat{y}_n)\,ds\Big)\Bigg]\\
  \geq & \ \limsup_{n \to \infty} \inf_{\theta\in \Theta}\frac{1}{p}\exp\Big(p\,T\, g^{\theta}(\hat{y}_n)\Big)
  \end{align*}
  and finally Lemma~\ref{le:approx-g} shows that
  \begin{align*}
 \limsup_{n \to \infty} \inf_{\theta\in \Theta}\frac{1}{p}\exp\Big(p\,T\, g^{\theta}(\hat{y}_n)\Big)
 = & \  \frac{1}{p}\exp\Big(p\,T\,\limsup_{n \to \infty} \sup_{y \in \sC \cap \sC^0_n}\inf_{\theta\in \Theta} g^{\theta}(y)\Big)\\
  = & \  \frac{1}{p}\exp\Big(p\,T\, \sup_{y \in \sC \cap \sC^0}\inf_{\theta\in \Theta} g^{\theta}(y)\Big),
  \end{align*}
  which proves the desired inequality. It remains to prove that
	\begin{equation*}
	\inf_{P \in \fP} \sup_{\pi \in \cA}  E^P[U(W_T^{\pi})]
	\leq  \sup_{\pi \in \cA}  \inf_{P \in \fP}   E^P[U(W_T^{\pi})].
	\end{equation*}
	 Indeed, by Lemma~\ref{le:proof-thm-limit-l1-pow}, it suffices to show that 
	 \begin{equation*}
	 \inf_{P \in \fP} \sup_{\pi \in \cA}E^P[U(W_T^{\pi})]\leq \frac{1}{p}\exp\Big(p\,T\, \sup_{y \in \sC \cap \sC^0}\inf_{\theta\in \Theta} g^{\theta}(y)\Big).
	 \end{equation*}
	We first notice that Lemma~\ref{le:levyvalue} implies
	\begin{equation*}
	\inf_{P \in \fP} \sup_{\pi \in \cA}E^P[U(W_T^{\pi})]\leq \inf_{P \in \fP_{L}} \sup_{\pi \in \cA}E^P[U(W_T^{\pi})]\leq \inf_{P \in \fP_{L}} \sup_{\pi \in \cA^P}E^P[U(W_T^{\pi})].
	\end{equation*}
	Using \cite[Theorem~3.2]{Nutz.09c}, we see that the right-hand side satisfies
	\begin{equation*}
	\inf_{P \in \fP_{L}} \sup_{\pi \in \cA^P}E^P[U(W_T^{\pi})]
	=\inf_{P \in \fP_{L}} \frac{1}{p}\exp\Big(p\,T\, \sup_{y \in \sC \cap \sC^0} g^{\theta^P}(y)\Big),
	\end{equation*}
	while the local minimax result of Proposition~\ref{prop:local-minimax} and the definition of $\fP_{L}$ yield that
	\begin{align*}
	 \inf_{P \in \fP_{L}} \frac{1}{p}\exp\Big(p\,T\, \sup_{y \in \sC \cap \sC^0} g^{\theta^P}(y)\Big)
	=  \frac{1}{p}\exp\Big(p\,T\, \sup_{y \in \sC \cap \sC^0}\inf_{\theta\in \Theta} g^{\theta}(y)\Big).
	\end{align*} 
	This completes the proof.
\end{proof}

The proof for part (iii) of Theorem~\ref{thm} is analogous to Lemma~\ref{le:proof-thm-l2-log} for the logarithmic case. We omit the details and proceed with part~(i) of Theorem~\ref{thm-compact}.

\begin{lemma} \label{le:proof-thm-limit-l4-pow}
	Assume that $\Theta$ is compact. Let $( \hat{\theta},\hat{y}) \in \Theta \times\sC\cap\sC^0$  be a saddle point of the function $g^{\theta}(y)$ and let $\hat{P} \in \fP_{L}$ be the L\'evy law with  triplet~$\hat{\theta}$. Then $(\hat{P},\hat{y}) \in \fP_{L}\times \cA$ is a saddle point of $(P,\pi)\mapsto E^P[U(W^\pi_T)]$ on $\fP \times \cA$ and
	\begin{equation*}
	\sup_{\pi\in \cA} \inf_{P\in\fP} E^{P} [U(W^{\pi}_T)]
	= \frac{1}{p}\exp\Big(p\,T \, g^{\hat{\theta}}(\hat{y})\Big).
	\end{equation*}
	\end{lemma}
	\begin{proof}
	By Lemma~\ref{le:levyvalue} and \cite[Theorem~3.2]{Nutz.09c}, we have
	\begin{equation*}\label{eq:1-le:proof-thm-limit-l4-pow}
	\inf_{P \in \fP} \sup_{\pi \in \cA}E^P[U(W_T^{\pi})]
	\leq \ \sup_{\pi \in \cA^{\hat{P}}}E^{\hat{P}}[U(W_T^{\pi})]=\frac{1}{p}\exp\Big(p\,T\, \sup_{y \in \sC \cap \sC^0} g^{\hat{\theta}}(y)\Big).
	\end{equation*}
	Setting $\hat{y}_n= (1-\frac{1}{n}) \hat{y}\in\cA_{n}$, Lemma~\ref{le:approx-g} yields that
	\begin{align*}
	\frac{1}{p}\exp\Big(p\,T\, \sup_{y \in \sC \cap \sC^0} g^{\hat{\theta}}(y)\Big)
	&=\frac{1}{p}\exp\Big(p\,T\,\inf_{\theta\in \Theta} g^{\theta}(\hat{y})\Big) \\
	&= \frac{1}{p}\exp\Big(p\,T\,\lim\limits_{n \to \infty}\inf_{\theta\in \Theta} g^{\theta}(\hat{y}_n)\Big)\\
	&= \lim\limits_{n \to \infty} \inf_{\theta\in \Theta} \frac{1}{p}\exp\Big(p\,T\,g^{\theta}(\hat{y}_n)\Big),
	\end{align*}
	and we deduce from Lemma~\ref{le:expo-martingale} that
	\begin{align*}
	& \ \lim\limits_{n \to \infty} \inf_{\theta\in \Theta} \frac{1}{p}\exp\Big(p\,T\,g^{\theta}(\hat{y}_n)\Big)\\
	= & \ \lim\limits_{n \to \infty} \inf_{P \in \fP}\Bigg( E^P\Bigg[\frac{U(W^{\hat{y}_n}_T)}{\frac{1}{p}\exp\big(p\,\int_0^T g^{\theta^P_s}(\hat{y}_n)\,ds\big)}\Bigg]\, \inf_{\theta\in \Theta} \frac{1}{p}\exp\Big(p\,T\,g^{\theta}(\hat{y}_n)\Big)\Bigg)
	\\
	\leq & \ \lim\limits_{n \to \infty} \inf_{P\in\fP} E^{P}[U(W_T^{\hat{y}_n})]\\
	\leq & \ \inf_{P\in\fP} \limsup_{n \to \infty} E^{P}[U(W_T^{\hat{y}_n})].
	\end{align*} 
	As Lemma~\ref{le:stoch-expo-limit} shows that
	\begin{equation*}
	\inf_{P\in\fP} \limsup_{n \to \infty} E^{P}[U(W_T^{\hat{y}_n})] 
	\leq \inf_{P\in\fP}  E^{P}[U(W^{\hat{y}}_T)] 
	\leq \sup_{\pi \in \cA}\inf_{P \in \fP} E^P[U(W_T^{\pi})],
	\end{equation*}
	all the above inequalities are equalities, and the result follows.
\end{proof}

Finally, the argument for part~(ii) of Theorem~\ref{thm-compact} is quite similar to Lemma~\ref{le:proof-thm-l4-log} and therefore omitted. This completes the proofs of Theorems~\ref{thm} and~\ref{thm-compact} for power utility.

\newcommand{\dummy}[1]{}


\begin{thebibliography}{10}

\bibitem{BertsekasShreve.78}
D.~P. Bertsekas and S.~E. Shreve.
\newblock {\em Stochastic Optimal Control. The Discrete-Time Case.}
\newblock Academic Press, New York, 1978.

\bibitem{BiaginiPinar.15}
S.~Biagini and M.~Pinar.
\newblock The robust {M}erton problem of an ambiguity averse investor.
\newblock {\em Preprint arXiv:1502.02847v1}, 2015.

\bibitem{Bogachev.07volII}
V.~I. Bogachev.
\newblock {\em Measure Theory. {V}ol. {II}}.
\newblock Springer-Verlag, Berlin, 2007.

\bibitem{DenisKervarec.13}
L.~Denis and M.~Kervarec.
\newblock Optimal investment under model uncertainty in nondominated models.
\newblock {\em SIAM J. Control Optim.}, 51(3):1803--1822, 2013.

\bibitem{FouquePunWong.14}
J.-P. Fouque, C.~S. Pun, and H.~Y. Wong.
\newblock Portfolio optimization with ambiguous correlation and stochastic
  volatilities.
\newblock {\em Preprint, University of Santa Barbara}, 2014.

\bibitem{HuPeng.09levy}
M.~Hu and S.~Peng.
\newblock {$G$}-{L}{\'e}vy processes under sublinear expectations.
\newblock {\em Preprint arXiv:0911.3533v1}, 2009.

\bibitem{JacodShiryaev.03}
J.~Jacod and A.~N. Shiryaev.
\newblock {\em Limit Theorems for Stochastic Processes}.
\newblock Springer, Berlin, 2nd edition, 2003.

\bibitem{Kallsen.00}
J.~Kallsen.
\newblock Optimal portfolios for exponential {L\'{e}}vy processes.
\newblock {\em Math. Methods Oper. Res.}, 51(3):357--374, 2000.

\bibitem{KallsenMuhleKarbe.08}
J.~Kallsen and J.~Muhle-Karbe.
\newblock Utility maximization in affine stochastic volatility models.
\newblock {\em Int. J. Theor. Appl. Finance}, 13(3):459--477, 2010.

\bibitem{Kardaras.09}
C.~Kardaras.
\newblock No-free-lunch equivalences for exponential {L}\'evy models under
  convex constraints on investment.
\newblock {\em Math. Finance}, 19(2):161--187, 2009.

\bibitem{LinRiedel.14}
Q.~Lin and F.~Riedel.
\newblock Optimal consumption and portfolio choice with ambiguity.
\newblock {\em Preprint arXiv:1401.1639v1}, 2014.

\bibitem{MatoussiPossamaiZhou.12utility}
A.~Matoussi, D.~Possamai, and C.~Zhou.
\newblock Robust utility maximization in non-dominated models with {2BSDEs}:
  The uncertain volatility model.
\newblock {\em Math. Finance}, 25(2):258--287, 2015.

\bibitem{Merton.69}
R.~C. Merton.
\newblock Lifetime portfolio selection under uncertainty: the continuous-time
  case.
\newblock {\em Rev. Econom. Statist.}, 51:247--257, 1969.

\bibitem{NeufeldNutz.13a}
A.~Neufeld and M.~Nutz.
\newblock Measurability of semimartingale characteristics with respect to the
  probability law.
\newblock {\em Stochastic Process. Appl.}, 124(11):3819--3845, 2014.

\bibitem{NeufeldNutz.13b}
A.~Neufeld and M.~Nutz.
\newblock Nonlinear {L\'e}vy processes and their characteristics.
\newblock {\em To appear in Trans. Amer. Math. Soc.}, 2014.

\bibitem{Nutz.09a}
M.~Nutz.
\newblock The opportunity process for optimal consumption and investment with
  power utility.
\newblock {\em Math. Financ. Econ.}, 3(3):139--159, 2010.

\bibitem{Nutz.09b}
M.~Nutz.
\newblock The {B}ellman equation for power utility maximization with
  semimartingales.
\newblock {\em Ann. Appl. Probab.}, 22(1):363--406, 2012.

\bibitem{Nutz.09c}
M.~Nutz.
\newblock Power utility maximization in constrained exponential {L}\'evy
  models.
\newblock {\em Math. Finance}, 22(4):690--709, 2012.

\bibitem{Nutz.13util}
M.~Nutz.
\newblock Utility maximization under model uncertainty in discrete time.
\newblock {\em Math. Finance}, 26(2):252--268, 2016.

\bibitem{Pollard.13}
D.~Pollard.
\newblock {\em Asymptopia}.
\newblock 2001.
\newblock
  \url{http://www.stat.yale.edu/~pollard/Books/Asymptopia/Appendices.pdf}.

\bibitem{Prohorov.56}
Yu.~V. Prohorov.
\newblock Convergence of random processes and limit theorems in probability
  theory.
\newblock {\em Teor. Veroyatnost. i Primenen.}, 1:177--238, 1956.

\bibitem{Quenez.04}
M.-C. Quenez.
\newblock Optimal portfolio in a multiple-priors model.
\newblock In {\em Seminar on {S}tochastic {A}nalysis, {R}andom {F}ields and
  {A}pplications {IV}}, volume~58 of {\em Progr. Probab.}, pages 291--321.
  Birkh{\"a}user, Basel, 2004.

\bibitem{Samuelson.69}
P.~A. Samuelson.
\newblock Lifetime portfolio selection by dynamic stochastic programming.
\newblock {\em Rev. Econ. Statist.}, 51(3):239--246, 1969.

\bibitem{Schied.06}
A.~Schied.
\newblock Risk measures and robust optimization problems.
\newblock {\em Stoch. Models}, 22(4):753--831, 2006.

\bibitem{Sion.58}
M.~Sion.
\newblock On general minimax theorems.
\newblock {\em Pacific J. Math.}, 8:171--176, 1958.

\bibitem{StoikovZariphopoulou.05}
S.~Stoikov and T.~Zariphopoulou.
\newblock Dynamic asset allocation and consumption choice in incomplete
  markets.
\newblock {\em Australian Econ. Pap.}, 44(4):414--454, 2005.

\bibitem{TalayZheng.02}
D.~Talay and Z.~Zheng.
\newblock Worst case model risk management.
\newblock {\em Finance Stoch.}, 6(4):517--537, 2002.

\bibitem{TevzadzeToronjadzeUzunashvili.13}
R.~Tevzadze, T.~Toronjadze, and T.~Uzunashvili.
\newblock Robust utility maximization for a diffusion market model with
  misspecified coefficients.
\newblock {\em Finance Stoch.}, 17(3):535--563, 2013.

\end{thebibliography}
\end{document}